\documentclass[a4paper,UKenglish,cleveref, autoref]{lipics-v2019}

\usepackage{float}
\usepackage{amsmath}
\usepackage{amsfonts}
\usepackage{amssymb}
\usepackage{xspace}
\usepackage[table,xcdraw]{xcolor}
\usepackage{graphicx}
\usepackage{multicol}

\usepackage{thm-restate}

\usepackage[shortcuts]{extdash}

\usepackage{tikz}

\usetikzlibrary{positioning}
\usetikzlibrary{decorations.pathreplacing}
\usetikzlibrary{automata,positioning}
\usetikzlibrary{decorations.pathmorphing}
\usetikzlibrary{decorations.markings}
\usetikzlibrary{decorations}
\usetikzlibrary{arrows}
\usetikzlibrary{patterns}
\usetikzlibrary{calc}
\usetikzlibrary{shapes}

\tikzstyle{ubrace} = [draw, thick, decoration={brace, amplitude=7pt, mirror, raise=0.0cm}, decorate,
    every node/.style={anchor=north, yshift=-0.15cm}]
\tikzstyle{rbrace} = [draw, thick, decoration={brace, amplitude=7pt, mirror, raise=0.0cm}, decorate,
    every node/.style={anchor=west, xshift= 0.15cm}]

\tikzstyle{obrace} = [draw, thick, decoration={brace, amplitude=7pt, raise=0.0cm}, decorate,
    every node/.style={anchor=south, yshift= 0.15cm}]
\tikzstyle{lbrace} = [draw, thick, decoration={brace, amplitude=7pt, raise=0.0cm}, decorate,
    every node/.style={anchor=east, xshift=-0.15cm}]

\tikzstyle{eve} = [draw, diamond, minimum size=20pt, inner sep=1pt]
\tikzstyle{adam} = [draw, regular polygon, regular polygon sides=4, inner sep=1pt]
\tikzstyle{state} = [circle, draw, inner sep=1pt, minimum size=20pt]
\tikzstyle{bool} = [circle, draw, inner sep=1pt, minimum size=5pt, scale=0.9]
\tikzstyle{dot} = [fill, circle, inner sep=0pt, minimum size=4pt]
\tikzstyle{dots} = [scale=2.0]
\tikzstyle{trans} = [draw, -{latex}, auto]
\tikzstyle{booltrans} = [draw, -{latex}, auto]
\tikzstyle{edge} = [draw, thick]
\tikzstyle{edgeBox} = [draw, thick,-latex]

\newcommand{\evalInt}[2]{%
	\pgfmathparse{int(#2)}%
	{\global\edef#1{\pgfmathresult}}%
}

\definecolor{darkgreen}{rgb}{0.13, 0.55, 0.13}
\newcommand{\todo}[1]{\textcolor{red}{\textsf{To do: #1}}}

\newcommand{\NotNeeded}[1]{}
\newcommand{\FullVersion}[1]{}

\newcommand{\Subject}[1]{\subparagraph{#1.}}
\newcommand{\ProofSketch}{\begin{proof}[Proof sketch.]}

\newcommand{\tuple}[1]{( #1  )}
\newcommand{\pair}{\tuple}

\newcommand{\Nat}{\ensuremath{\mathbb{N}}\xspace}

\newcommand{\Func}[1]{{\mathsf{#1}}}
\newcommand{\BP}{\Func{B}^+}

\newcommand{\pro}[2]{\ensuremath{#1 \times #2}\xspace} 

\newcommand{\strat}{\stratE}
\newcommand{\boxes}{\mathsf{B}}

\newcommand{\A}{{\cal A}}
\newcommand{\B}{{\cal B}}
\newcommand{\C}{{\cal C}}
\newcommand{\D}{{\cal D}}
\newcommand{\E}{{\cal E}}

\newcommand{\G}{{\cal G}}

\newcommand{\N}{{\cal N}}

\renewcommand{\S}{{\cal S}}

\newcommand{\restr}{{\upharpoonright}}

\newcommand{\trans}[3]{#1\xrightarrow[]{#2}#3}

\newcommand{\EGFG}{\exists\mathrm{\text{-}GFG}}
\newcommand{\AGFG}{\forall\mathrm{\text{-}GFG}}

\newcommand{\boxX}[1]{\mathtt{box}(#1)}
\newcommand{\boxA}{\boxX{\A}} 
\newcommand{\boxAco}{\boxX{\bar \A}}

\newcommand{\noLab}{\epsilon}

\newcommand{\stratE}{\sigma}
\newcommand{\stratA}{\tau}
\newcommand{\upd}{\mathit{upd}}
\newcommand{\letter}{a}

\newcommand{\pspace}{{\sc{PSpace}}\xspace}
\newcommand{\ptime}{{\sc{PTime}}\xspace}
\newcommand{\exptime}{{\sc{Exptime}}\xspace}

\title{On the Succinctness of Alternating Parity Good-for-Games Automata (Full Version\footnote{This is the full version of the paper of the same name published at FSTTCS 2020 
})}

\author{Udi Boker}{Interdisciplinary Center (IDC) Herzliya, Israel}{udiboker@gmail.com}{}{Israel Science Foundation grant 1373/16}

\author{Denis Kuperberg}{CNRS, LIP, \'Ecole Normale Supérieure, Lyon, France}{denis.kuperberg@ens-lyon.fr}{0000-0001-5406-717X}{}

\author{Karoliina Lehtinen}{University of Liverpool, United Kingdom}{k.lehtinen@liverpool.ac.uk}{0000-0003-1171-8790}{
This project has received funding from the European Union’s Horizon 2020 research and innovation programme under the Marie Skłodowska-Curie grant
agreement No 892704.
}

\author{Micha\l~Skrzypczak}{Institute of Informatics, University of Warsaw, Poland}{mskrzypczak@mimuw.edu.pl}{0000-0002-9647-4993}{}

\authorrunning{U. Boker, D. Kuperberg, K. Lehtinen, M. Skrzypczak}

\Copyright{Udi Boker, Denis Kuperberg, Karoliina Lehtinen, Micha\l~Skrzypczak}



\category{}

\supplement{}

\acknowledgements{}

\nolinenumbers 

\hideLIPIcs  

\EventEditors{John Q. Open and Joan R. Access}
\EventNoEds{2}
\EventLongTitle{42nd Conference on Very Important Topics (CVIT 2016)}
\EventShortTitle{CVIT 2016}
\EventAcronym{CVIT}
\EventYear{2016}
\EventDate{December 24--27, 2016}
\EventLocation{Little Whinging, United Kingdom}
\EventLogo{}
\SeriesVolume{42}
\ArticleNo{23}

\begin{document}
\maketitle

\begin{abstract}
We study alternating parity good-for-games (GFG) automata, i.e., alternating parity automata where both conjunctive and disjunctive choices can be resolved in an online manner, without knowledge of the suffix of the input word still to be read.

We show that they can be exponentially more succinct than both their nondeterministic and universal counterparts.
Furthermore, we present a single exponential determinisation procedure and an \exptime upper bound to the problem of recognising whether an alternating automaton is GFG.

We also study the complexity of deciding ``half-GFGness'', a property specific to alternating automata that only requires nondeterministic choices to be resolved in an online manner. We show that this problem is \pspace-hard already for alternating automata on finite words.
\end{abstract}

\section{Introduction}
\label{sec:Introduction}

\emph{Good-for-games} (GFG) automata were first introduced in~\cite{HP06} as a tool for solving the synthesis problem.
The equivalent notion of \emph{history\=/determinism} was introduced independently in~\cite{Col09} in the context of regular cost functions.
Intuitively, a nondeterministic automaton is GFG if nondeterminism can be resolved on the fly, only with knowledge of the input word read so far.
GFG automata can be seen as an intermediate formalism between deterministic and nondeterministic ones, with advantages from both worlds. Indeed, like deterministic automata, GFG automata enjoy good compositional properties---useful for solving games and composing automata and trees---and easy inclusion checks~\cite{BKKS13}. Like nondeterministic automata, they can be exponentially more succinct than deterministic automata~\cite{KS15}.

In recent years, much effort has gone into understanding various properties of nondeterministic GFG automata, for instance their relationship with deterministic automata~\cite{BKKS13, KS15,BKS17,KM19}, applications in probabilistic model checking~\cite{KMBK14} and synthesis of LTL, $\mu$\=/calculus and context-free properties~\cite{IK19,LZ20}, decision procedures for GFGness~\cite{LR13,KS15,BK18}, minimisation~\cite{AK19}, and links with recent advances in parity games~\cite{CF19}.

Alternating GFG automata are a natural generalisation of nondeterministic GFG automata that enjoy the same compositional properties as nondeterministic GFG automata, while providing more flexibility. 
As we show in the present work, for some languages alternating GFG parity automata can also be exponentially more succinct, allowing for better synthesis procedures. Indeed, two-player games with winning conditions given by alternating GFG automata are solvable in quasipolynomial time, via a linear reduction to parity games, while for winning conditions given by arbitrary alternating automata, solving games requires determinisation and has therefore double-exponential complexity.

Alternating GFG automata were introduced independently by Colcombet~\cite{colcombet2013hab} and Quirl~\cite{Quirl} while a~form of alternating GFG automata with requirements specific to counters were also considered in~\cite{KV11}, as a tool to study cost functions on infinite trees.
Boker and Lehtinen studied the expressiveness and succinctness of alternating GFG automata in~\cite{BL19}, showing that they
\begin{itemize}
\item are not more succinct than DFAs on finite words,
\item are as expressive as deterministic ones of the same acceptance condition on infinite words,
\item and can be determinised with a $2^{\theta(n)}$ size blowup for the B\"uchi and coB\"uchi conditions.
\end{itemize}

Many questions about GFG alternating automata were left open, in particular
whether there exists a~doubly exponential succinctness gap between alternating GFG and deterministic automata, and the complexity of deciding whether an~alternating parity automaton is GFG. 



\Subject{Succinctness of alternating GFG automata}
We show that  there is a single exponential gap between alternating parity GFG automata and deterministic ones, thereby answering a~question left open in~\cite{BL19}. This is in contrast to general alternating automata, for which determinisation incurs a double-exponential size increase.
However, we also show that alternating GFG automata can present exponential succinctness compared to both nondeterministic and universal GFG automata. This means that alternating GFG automata can be used to reduce the complexity of solving some games with complex acceptance conditions.

\Subject{Recognising GFG automata}

We give an \exptime upper bound to the problem of deciding whether an~alternating parity automaton is GFG, matching the known upper bound for recognising nondeterministic parity GFG automata.

We also study the complexity of deciding ``half\=/GFGness'', i.e.,\ whether the nondeterminism (or universality) of an automaton is GFG. This property guarantees that composition with games preserves the winner for one of the players. We show that already on finite words, this problem is \pspace-hard, and it is in \exptime for alternating B\"uchi automata.
This shows that a \ptime algorithm for deciding GFGness must exploit the subtle interplay between nondeterminism and universality, and cannot be reduced to checking independently whether each of them is~GFG.

\paragraph*{\large Roadmap}
We begin with some definitions, after which, in~\cref{sec:Alternating-behaviour}, we define alternating GFG automata, study their succinctness and the complexity of deciding half\=/GFGness, that is, whether the nondeterminism within an alternating automaton is GFG.  \Cref{sec:Determinisation} provides a single-exponential determinisation procedure for alternating GFG parity automata. \Cref{sec:deciding} shows that GFGness of alternating parity automata is in \exptime, using the determinisation of the previous section.
Throughout the paper, we provide high-level proof sketches, with detailed technical developments 
in the appendix.

\section{Preliminaries}
\label{sec:Preliminaries}

\Subject{Words and automata}
An \emph{alphabet} $\Sigma$ is a finite nonempty set of letters. A finite (resp.\ infinite) \emph{word} $u=u_0 \ldots u_k\in \Sigma^{*}$ (resp.\ $w=w_0 w_1\ldots\in \Sigma^{\omega}$) is a finite (resp.\ infinite) sequence of letters from $\Sigma$. 
A \emph{language} is a set of words, and the empty word is written $\epsilon$. We denote a set $\{i,i+1,\ldots,j\}$ of integers by $[i,j]$.
 
An \emph{alternating word automaton} is a~tuple $\A=(\Sigma,Q,\iota,\delta,\alpha)$, where: $\Sigma$ is an alphabet; $Q$ is a finite nonempty set of states; $\iota\in Q$ is an~initial state; $\delta\colon Q\times \Sigma \to \BP(Q)$ is a transition function where $\BP(Q)$ is the set of positive Boolean formulas (\emph{transition conditions}) over $Q$; and $\alpha$, on which we elaborate below, is either an acceptance condition or a transition labelling on top of which an acceptance condition is defined.
For a state $q\in Q$, we denote by $\A^q$ the automaton that is derived from $\A$ by setting its initial state $\iota$ to $q$. 

An automaton $\A$ is nondeterministic (resp.\ universal) if all its transition conditions are disjunctions (resp.\ conjunctions), and it is deterministic if all its transition conditions are just states. We represent the transition function of nondeterministic and universal automata as $\delta\colon Q\times \Sigma\to 2^Q$, and of a deterministic automaton as $\delta\colon Q\times \Sigma\to Q$. A \emph{transition} of an automaton is a triple $(q,\letter,q')\in Q{\times}\Sigma{\times} Q$, sometimes also written $\trans{q}{\letter}{q'}$.

We denote by $\widehat{\delta}\subseteq \BP(Q)$ the set of all subformulas of formulas in the image of $\delta$, i.e., all the Boolean formulas that ``appear'' somewhere in the transition function of $\A$. 

\Subject{Acceptance conditions}
There are various acceptance (winning) conditions, defined with respect to the set of transitions\footnote{Acceptance is defined in the literature with respect to either states or transitions; for technical reasons we prefer to work with acceptance on transitions.} that a path of $\A$ visits infinitely often. (Notice that a transition condition allows for many possible transitions.)
We later formally define acceptance of a word $w$ by $\A$ in terms of games, and consider a path of $\A$ on a word $w$ as a play in that game. For nondeterministic automata, a ``run'' coincides with a ``path''.

Some of the acceptance conditions are defined on top of a labelling of the transitions rather than directly on the transitions. In particular, in the parity condition, we have $\alpha\colon Q\times\Sigma\times Q \to \Gamma$, where $\Gamma\subseteq\Nat$ is a finite set of priorities and a path is accepting if and only if the highest priority seen infinitely often on it is even.

The B\"uchi and coB\"uchi conditions are special cases of the parity condition with $\Gamma=\{1,2\}$ and $\Gamma=\{0,1\}$, respectively. When speaking of B\"uchi and coB\"uchi automata, we often refer to $\alpha$ as the set of ``accepting transitions'', namely the transitions that are mapped to $2$ in the B\"uchi case and to $0$ in the coB\"uchi case.
The weak condition is a special case of both the B\"uchi and coB\"uchi conditions, in which every path eventually remains in the same priority.

The Rabin and Streett conditions are more involved, yet defined directly on the set $T$ of transitions. A Rabin condition is a set $\{\pair{B_1,G_1}, \pair{B_2,G_2},\ldots, \pair{B_k,G_k}\}$, with $B_i, G_i \subseteq T$,
and a path $\rho$ is accepting iff for some $i\in [1,k]$, we have that the set $\inf(\rho)$ of transitions that are visited infinitely often in $\rho$ satisfies ($\inf(\rho) \cap B_i = \emptyset$ and $\inf(\rho) \cap G_i \neq \emptyset$). A Streett condition is dual: a set $\{\pair{B_1,G_1}, \pair{B_2,G_2},
\ldots, \pair{B_k,G_k}\}$, with $B_i, G_i \subseteq Q$, whereby a~path $\rho$ is accepting iff for all $i\in [1,k]$, we have ($\inf(\rho) \cap B_i = \emptyset$ or $\inf(\rho) \cap G_i \neq \emptyset$).

\Subject{Sizes and types of automata}
The size of $\A$ is the maximum of the alphabet size, the number of states, the transition function length, which is the sum of the transition condition lengths over all states and letters, and the acceptance condition's index, which is $1$ for weak, B\"uchi and coB\"uchi, $|\Gamma|$ for parity, and $k$ for Rabin and Street.

We sometimes abbreviate automata types by three-letter acronyms in $\{$D, N, U, A$\} \times \{$F, W, B, C, P, R, S$\} \times \{$A,W$\}$. The first letter stands for the transition mode, the second for the acceptance condition, and the third indicates that the automaton runs on finite or infinite words. For example, DPW stands for a deterministic parity automaton on infinite words.

\Subject{Games and strategies}
Some of our technical proofs use standard concepts of an arena, a game, a winning strategy, etc. For the sake of completeness, we provide precise mathematical definitions of these objects in \cref{ap:Preliminaries}. Here we will just overview the involved concepts.

First, we work with two\=/player games of perfect information, where the players are Eve and Adam. These games are played on graphs (called arenas). Most of the considered games are of infinite duration and their winning condition is expressed in terms of the infinite sequences of edges taken during the play. We invoke results of determinacy (one of the players has a winning strategy), as well as of \emph{positional determinacy} (one of the players has a strategy that depends only on the last position of the play).

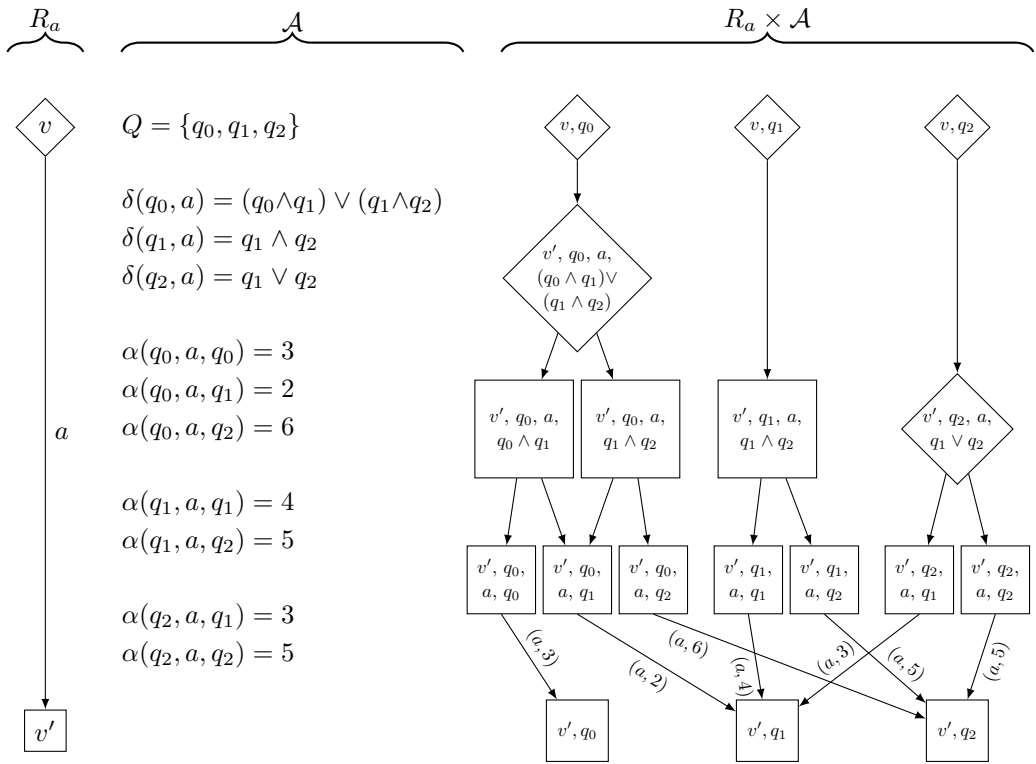
\begin{figure}[h]
	\centering
	\begin{tikzpicture}

\draw (-0.5, 1) edge[obrace] node{$R_a$} (0.5, 1);
\draw (1, 1) edge[obrace] node{$\A$} (5.5, 1);
\draw (6, 1) edge[obrace] node{$R_a\times \A$} (13, 1);

\node[eve,  minimum size=22pt] (vi) at (0, 0) {$v$};
\node[adam, minimum size=22pt] (vt) at (0,-8) {$v'$};
\draw[trans] (vi) edge node {$a$} (vt);

\tikzstyle{flow} = [anchor=mid west, inner sep=0cm]
\newcommand{\x}{1}

\node[flow] at (\x, -0.0) {$Q=\{q_0, q_1, q_2\}$};

\node[flow] at (\x, -1.0) {$\delta(q_0,a)=(q_0{\land} q_1)\lor(q_1{\land}q_2)$};
\node[flow] at (\x, -1.5) {$\delta(q_1,a)=q_1 \land q_2$};
\node[flow] at (\x, -2.0) {$\delta(q_2,a)=q_1 \lor q_2$};

\node[flow] at (\x, -3.0) {$\alpha(q_0,a,q_0)=3$};
\node[flow] at (\x, -3.5) {$\alpha(q_0,a,q_1)=2$};
\node[flow] at (\x, -4.0) {$\alpha(q_0,a,q_2)=6$};

\node[flow] at (\x, -5.0) {$\alpha(q_1,a,q_1)=4$};
\node[flow] at (\x, -5.5) {$\alpha(q_1,a,q_2)=5$};

\node[flow] at (\x, -6.5) {$\alpha(q_2,a,q_1)=3$};
\node[flow] at (\x, -7.0) {$\alpha(q_2,a,q_2)=5$};


\foreach \x/\p in {0/eve, 1/adam, 2/adam} {
	\node[eve,  scale=0.7, minimum size=35pt] (v\x) at (7 + 2.5*\x, +0) {$v, q_{\x}$};
	\node[adam, scale=0.7, minimum size=30pt] (f\x) at (7 + 2.5*\x, -8) {$v', q_{\x}$};
}

\tikzstyle{Atree} = [adam, scale=0.7, inner sep=0cm, node distance=0cm and 0cm, align=center]
\tikzstyle{Etwo} = [eve, scale=0.7, inner sep=0cm, node distance=0cm and 0cm, align=center]

\node[Etwo] (w0) at ($(v0)+(-0,-2)$) {$v'$, $q_0$, $a$, \\ $(q_0\land q_1)\lor$\\$(q_1\land q_2)$};
	
\tikzstyle{Atree} = [adam, scale=0.7, inner sep=0cm, node distance=0cm and 0cm, align=center]
\tikzstyle{Etree} = [eve, scale=0.7, inner sep=0cm, node distance=0cm and 0cm, align=center]

\node[Atree] (wl0) at ($(v0)+(-0.7,-4)$) {$v'$, $q_0$, $a$,\\$q_0\land q_1$};

\node[Atree] (wr0) at ($(v0)+(+0.7,-4)$) {$v'$, $q_0$, $a$,\\$q_1\land q_2$};

\node[Atree] (wl1) at ($(v1)+(+0.0,-4)$) {$v'$, $q_1$, $a$,\\$q_1\land q_2$};

\node[Etree] (wr2) at ($(v2)+(-0,-4)$) {$v'$, $q_2$, $a$,\\$q_1\lor q_2$};

\tikzstyle{Afour} = [adam, scale=0.7, inner sep=0cm, node distance=0cm and 0cm, align=center]
\tikzstyle{Efour} = [eve, scale=0.7, inner sep=0cm, node distance=0cm and 0cm, align=center]

\node[Afour] (f00) at ($(v0)+(-1,-6)$) {$v'$, $q_0$,\\$a$, $q_0$};
\node[Afour] (f01) at ($(v0)+(+0,-6)$) {$v'$, $q_0$,\\$a$, $q_1$};
\node[Afour] (f02) at ($(v0)+(+1,-6)$) {$v'$, $q_0$,\\$a$, $q_2$};

\node[Afour] (f11) at ($(wl1)+(-0.25,-2)$) {$v'$, $q_1$,\\$a$, $q_1$};
\node[Afour] (f12) at ($(wl1)+(+0.75,-2)$) {$v'$, $q_1$,\\$a$, $q_2$};

\node[Afour] (f21) at ($(wr2)+(-0.5,-2)$) {$v'$, $q_2$,\\$a$, $q_1$};
\node[Afour] (f22) at ($(wr2)+(+0.5,-2)$) {$v'$, $q_2$,\\$a$, $q_2$};

\draw[trans] (v0) edge (w0);	
\draw[trans] (v1) edge (wl1);	
\draw[trans] (v2) edge (wr2);	

\draw[trans] (w0) edge (wl0);
\draw[trans] (w0) edge (wr0);

\draw[trans] (wl0) edge (f00);
\draw[trans] (wl0) edge (f01);
\draw[trans] (wr0) edge (f01);
\draw[trans] (wr0) edge (f02);

\draw[trans] (wl1) edge (f11);
\draw[trans] (wl1) edge (f12);

\draw[trans] (wr2) edge (f21);
\draw[trans] (wr2) edge (f22);

\tikzstyle{edtr} = [sloped, pos=0.5, anchor=south, scale=0.7]

\draw (f00.south) edge[trans] node[edtr, above] {$(a,3)$} (f0);
\draw (f01.south) edge[trans] node[edtr, below] {$(a,2)$} (f1);
\draw (f02.south) edge[trans] node[edtr, pos=0.15, below] {$(a,6)$} (f2);

\draw (f11.south) edge[trans] node[edtr, pos=0.7, below] {$(a,4)$} (f1);
\draw (f12.south) edge[trans] node[edtr, pos=0.7, above] {$(a,5)$} (f2);

\draw (f21.south) edge[trans] node[edtr, pos=0.61, above] {$(a,3)$} (f1);
\draw (f22.south) edge[trans] node[edtr, below] {$(a,5)$} (f2);
\end{tikzpicture}
	\caption{A one-step arena over a letter $a\in\Sigma$, obtained as a product of a~simple arena $R_a$ with the alternating parity automaton $\A$. In this example $v$ is controlled by Eve and $v'$ by Adam. The transitions with no label are labelled by $\noLab$. Diamond\=/shaped positions belong to Eve and square\=/shaped positions belong to Adam.}
	\label{fig:one-step-arena}
\end{figure}

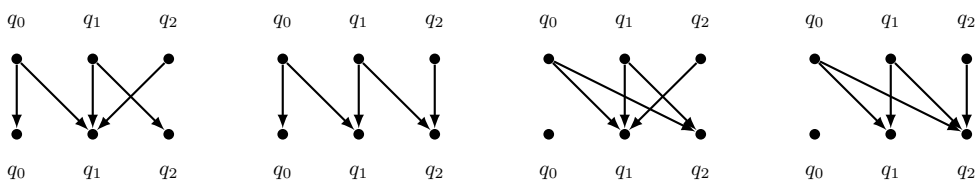
\begin{figure}
	\centering
	\newcommand{\drawBox}[2]{
\foreach \x in {0,1,2} {
  \node[scale=0.8] at (#1+\x, #2+0.5) {$q_\x$};
  \node[dot] (u\x) at (#1+\x, #2+0) {};
  \node[dot] (l\x) at (#1+\x, #2-1) {};
  \node[scale=0.8] at (#1+\x, #2-1.5) {$q_\x$};
}
}

\begin{tikzpicture}
\drawBox{0}{0}
\draw[edgeBox] (u0) -- (l0);
\draw[edgeBox] (u0) -- (l1);

\draw[edgeBox] (u1) -- (l1);
\draw[edgeBox] (u1) -- (l2);

\draw[edgeBox] (u2) -- (l1);

\drawBox{3.5}{0}
\draw[edgeBox] (u0) -- (l0);
\draw[edgeBox] (u0) -- (l1);

\draw[edgeBox] (u1) -- (l1);
\draw[edgeBox] (u1) -- (l2);

\draw[edgeBox] (u2) -- (l2);

\drawBox{7}{0}
\draw[edgeBox] (u0) -- (l1);
\draw[edgeBox] (u0) -- (l2);

\draw[edgeBox] (u1) -- (l1);
\draw[edgeBox] (u1) -- (l2);

\draw[edgeBox] (u2) -- (l1);

\drawBox{10.5}{0}
\draw[edgeBox] (u0) -- (l1);
\draw[edgeBox] (u0) -- (l2);

\draw[edgeBox] (u1) -- (l1);
\draw[edgeBox] (u1) -- (l2);

\draw[edgeBox] (u2) -- (l2);
\end{tikzpicture}
	\caption{The four possible boxes corresponding to Eve's choices in the one-step arena of \cref{fig:one-step-arena}. (All edges should be labelled with $a$, which we omit for better readability.)}
	\label{fig:boxes}
\end{figure}

\Subject{Model-checking games}

To represent the semantics of an alternating automaton $\A$, we treat the Boolean formulas that appear in the transition conditions of $\A$ as games. More precisely, given a letter $\letter\in\Sigma$ we represent the transition conditions $q\mapsto\delta(q,\letter)\in\BP(Q)$ as the \emph{one-step arena} over $\letter$ (see \cref{fig:one-step-arena}). A play over this arena begins in a state $q\in Q$; then the players go down the formula $\delta(q,a)$ with Eve resolving disjunctions and Adam resolving conjunctions; and finally they reach an atom $q'\in Q$ and the play stops. This means that a play over the one-step arena over $\letter$ results in a transition of the form $\trans{q}{\letter}{q'}$. 

The language $L(\A)$ of an alternating automaton $\A$ over an alphabet $\Sigma$ is defined via the \emph{model-checking game}, defined for an~automaton $\A$ and a~word $w=a_0a_1a_2\dots\in \Sigma^\omega$. A \emph{configuration} of this game is a state $q$ of $\A$ and a~position $i\in\omega$ of $w$, starting at $(\iota,0)$. In the $i$th round, starting at configuration $(q_i,i)$, the players play on the one-step arena from $q_i$ over $a_i$, resulting in a transition $\trans{q_i}{a_i}{q_{i+1}}$. The next configuration is $(q_{i+1},i{+}1)$. The acceptance condition of $\A$ becomes the winning condition of this game. $\A$ \emph{accepts} $w$ if Eve has a winning strategy in this game. 

For technical convenience, we define (in \cref{ap:Preliminaries}) the model-checking game in terms of a \emph{synchronised product} of the word $w$ (treated as an infinite graph) and the automaton $\A$. Synchronised products turn out to be useful in the analysis of various games presented in this paper and will be used throughout the technical versions of the proofs, in the appendix.

\begin{definition}
\label{def:auto-compl}
Given an alternating automaton $\A$, we denote by $\overline{\A}$ the \emph{dual} automaton: it has the same alphabet, set of states, and initial state. Its transition conditions $\delta_{\overline{\A}}(q,\letter)$ are obtained from those of $\A$ by replacing each disjunction ${\lor}$ with conjunction ${\land}$ and vice versa. Its acceptance condition is the dual of $\A's$ condition. (In parity automata, all priorities are increased by $1$.) $\overline{\A}$ recognises the complement $L(\A)^\mathrm{c}$ of $L(\A)$.
\end{definition}

\Subject{Boxes}
Another technical concept that we use is that of \emph{boxes} (see \cref{fig:boxes}), which describe Eve's \emph{local} strategies for resolving disjunctions within a transition condition. Consider an alternating automaton $\A$ and a letter $\letter\in\Sigma$. Moreover, fix a strategy $\stratE$ of Eve that resolves disjunctions in all the transition conditions $\delta(q,\letter)$ for $q\in Q$. Now, the \emph{box} of $\A$, $\letter$, and $\stratE$ is a subset of $Q\times \Sigma\times Q$ and contains a triple $(q,\letter,q')$ iff $\stratE$ resolves disjunctions of $\delta(q,\letter)$ in such a way that Adam (resolving conjunctions) can reach the atom $q'$. In other words, this box contains $(q,\letter,q')$ if there is a play consistent with $\stratE$ on $\delta(q,\letter)$ that reaches the atom $q'$. We use $\beta$ to denote single boxes and by $\boxes_{\A,\letter}$ we denote the set of all boxes of $\A$ and $\letter$, while $\boxes_\A$ denotes the union $\bigcup_{\letter\in\Sigma} \boxes_{\A,\letter}$.
We give a more formal definition based on synchronised products in the Appendix, see page~\pageref{def:boxes}.

\begin{definition}
\label{def:path-in-boxes}
Given a~sequence of boxes $\pi=b_0,b_1,\ldots$ of an automaton $\A$ and a~path $\rho=(q_0,\letter_0,q_1),(q_1,\letter_1,q_2),\ldots$, we say that $\rho$ is a~\emph{path of $\pi$} if for every $i$ we have $(q_i,\letter_i,q_{i+1})\in b_i$.
The sequence $\pi$ is said to be~\emph{universally accepting} if every path in $\pi$ is accepting in $\A$.
\end{definition}

Intuitively, a sequence of boxes $\pi$ as above represents a particular positional strategy $\strat$ of Eve in the model-checking game over the word $w=a_0a_1a_2\ldots$ In that case, a path of $\pi$ corresponds to a possible play of this game consistent with $\strat$, and the sequence is universally accepting if and only if the strategy is winning.

\section{Good-For-Games Alternating Automata}
\label{sec:Alternating-behaviour}

Good-for-games (GFG) nondeterministic automata are automata in which the nondeterministic choices can be resolved without looking at the future of the word. For example, consider an automaton that consists of a nondeterministic choice between a component that accepts words in which $a$ occurs infinitely often and a component that accepts words in which $a$ occurs finitely often. This automaton accepts all words but is not GFG since the nondeterministic choice of component cannot be resolved without knowing the whole word.

To extend this definition to alternating automata, we must look both at its nondeterminism and universality and require that both can be resolved without knowledge of the future. The following letter games capture this intuition.

\begin{figure}
\begin{center}
 \begin{tikzpicture}[->,>=stealth',shorten >=1pt,auto,node distance=2.8cm,semithick]
  \tikzstyle{every state}=[text=black]

  \node[initial,state] (A)              {$\vee$};
  \node[state,accepting] (B) [right of=A] {$\wedge$};
  \node[state] (C) [right of=B] {};
  \node[state, accepting] (D) [right of=C] {};

  \path (A) edge [] node [above] {$b,c$} (B)
  		(B) edge [] node [above] {$b$} (C)
  		(C) edge [] node [above] {$c$} (D)
  		(A) edge [loop above] node {$a,b,c$} (A)
  		(B) edge [loop above] node {$b,c$} (B)
  	 	(C) edge [loop above] node {$b$} (C)
        (D) edge [loop above] node {$b,c$} (D);
\end{tikzpicture}
\caption{Alternating weak automaton accepting words over $\{a,b,c\}$ in which $a$ occurs finitely often and $c$ occurs infinitely often. Omitted transitions lead to a rejecting sink.}\label{ex:alt}
\end{center}
\end{figure}
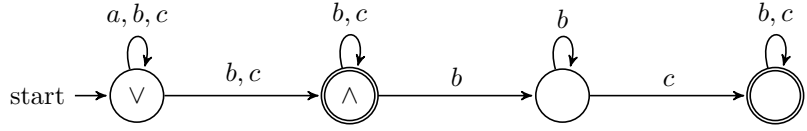

\begin{definition}[Letter games~\cite{BL19}]
\label{def:LetterGames}
Given an alternating automaton $\A$, Eve's letter game proceeds at each turn from a state $q$ of $\A$, starting from the initial state of $\A$, as follows:
\begin{itemize}
\item Adam chooses a letter $\letter$,
\item Adam and Eve play on the one-step arena over $\letter$ from $q$ to a new state $q'$, where Eve resolves disjunctions and Adam conjunctions.
\end{itemize}
A play of the letter game thus generates a word $w$ and a path $\rho$ of $\A$ on $w$. Eve wins this play if either $w\notin L(\A)$ or $\rho$ is accepting in $\A$.

Adam's letter game is similar, except that Eve chooses letters and Adam wins if either $w\in L(\A)$ or the path $\rho$ is rejecting.
\end{definition}
A more formal definition is given in \cref{ap:Alternating}.

\begin{definition}[GFG automata~\cite{BL19}]
An automaton $\A$ is \emph{$\EGFG$} if Eve wins her letter game; it is \emph{$\AGFG$} if Adam wins his letter game. Finally, $\A$ is \emph{GFG} if it is both $\EGFG$ and $\AGFG$.
\end{definition}

As shown in \cite[Theorem~8]{BL19}, an automaton $\A$ is GFG if and only if it is indeed ``good for playing games'', in the sense that its product with every game whose winning condition is $L(\A)$ preserves the winner of the game.

\begin{example}
The automaton in~\Cref{ex:alt} accepts the language $L$ of words in which $a$ occurs finitely often and $c$ occurs infinitely often. Here Eve loses her letter game: Adam can play $c$ until Eve takes the transition to the second state, and then play $a$ followed by $c^\omega$. 
Conversely, Eve wins Adam's letter game: her strategy is to play $b$, take the transition to the second state an keep playing $b$ until Adam takes the transition into the third state, after which she plays $c$ once and then $b^\omega$. 
 This automaton is neither $\EGFG$ nor $\AGFG$, and taking its product with games with $L$ as winning condition does not preserve the winner of the game.
 
In contrast, the automaton in~\Cref{ex:alt-gfg} is $\AGFG$ but not $\EGFG$. Indeed, Adam's winning strategy in his letter game is to resolve the conjunction from the middle state by always moving to the right-hand state when Eve plays $b$. This forces Eve to choose between playing $c$ infinitely many times (in which case, the word is in the language) or letting Adam build a rejecting run. Taking its product with \emph{one-player games} with winning condition $L$ preserves the winner whenever Eve is the player controlling all positions. However, this is not the case for one-player games where Adam is the sole player.

\end{example}

\begin{figure}
\begin{center}
 \begin{tikzpicture}[->,>=stealth',shorten >=1pt,auto,node distance=2.8cm,semithick]
  \tikzstyle{every state}=[text=black]

  \node[initial,state] (A)              {$\vee$};
  \node[state,accepting] (B) [right of=A] {$\wedge$};
  \node[state] (C) [right of=B] {};

  \path (A) edge [] node [above] {$b,c$} (B)
  		(B) edge [] node [above] {$b$} (C)
  		(C) edge [bend left] node [below] {$c$} (B)
  		(A) edge [loop above] node {$a,b,c$} (A)
  		(B) edge [loop above] node {$b,c$} (B)
  	 	(C) edge [loop above] node {$b$} (C);
\end{tikzpicture}
\caption{Alternating coB\"uchi $\AGFG$ automaton accepting words over $\{a,b,c\}$ in which $a$ occurs finitely often and $c$ occurs infinitely often.}\label{ex:alt-gfg}
\end{center}
\end{figure}
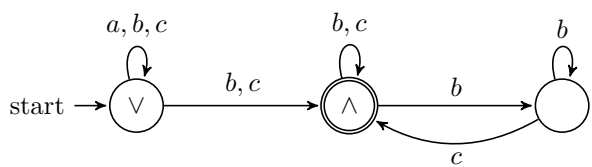
\subsection{Alternating GFG vs. Nondeterministic and Universal Ones}


We show in this section that alternating GFG automata can be more succinct than both nondeterministic and universal GFG automata.

\begin{lemma}
\label{lem:Cn_family}
There is a family $(\C_n)_{n\in\Nat}$ of alternating GFG $\{0,1,2\}$-parity automata of size linear in $n$ over a fixed alphabet, such that every nondeterministic GFG parity automaton and universal GFG parity automaton for $L(\C_n)$ is of size $2^{\Omega(n)}$.
\end{lemma}

\begin{proof}

From~\cite{KS15}, there is a family $(\A_n)_{n\in\Nat}$ of GFG-NCWs with $n$ states over a fixed alphabet $\Sigma$, such that every DPW for $L_n = L(\A_n)$ is of size $2^{\Omega(n)}$. 
For every $n\in\Nat$, let $\B_n$ be the dual of $\A_n$, so $\B_n$ is a UBW accepting $\overline{L_n}$. 
We build an APW $\C_n$ over $\Sigma$ of size linear in $n$, 
by setting its initial state to move to the initial state of $\A_n$ when reading the letter $a\in\Sigma$ and to the initial state of $\B_n$ when reading the letter $b\in\Sigma$. 
The acceptance condition of $\C_n$ is a parity condition with priorities $\{0,1,2\}$: accepting transitions of $\A_n$ are assigned priority $0$, and accepting transitions of $\B_n$ priority $2$. Other transitions have priority $1$.

The automaton $C_n$ is represented below:

\begin{center}
\includegraphics[scale=.5]{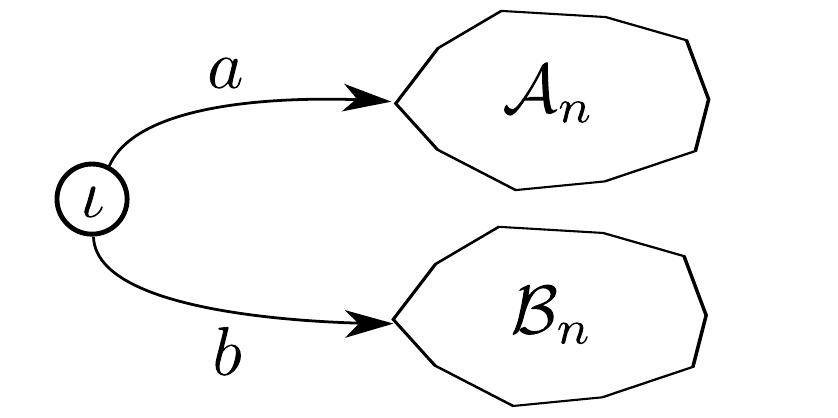}
\end{center}

Observe that $L(\C_n) = a L_n \cup b \overline{L_n}$, and that $\C_n$ is GFG: its initial state has only deterministic transitions, and over the $\A_n$ and $\B_n$ components, the strategy to resolve the nondeterminism and universality, respectively, follows the strategy to resolve the nondeterminism of $\A_n$, which is guaranteed due to $\A_n$'s GFGness.

Consider a GFG UPW $\E_n$ for $L(\C_n)$, and let $q$ be a state to which $\E_n$ moves when reading $a$, according to some strategy that witnesses $\E_n$'s GFGness. Then $\E_n^q$ is a GFG UPW for $L_n$.  Its dual is therefore a GFG NPW $\E'_n$ for $\overline{L_n}$.

Since $\A_n$ is a GFG NPW for $L_n$, by~\cite[Theorem~4]{BKKS13} we obtain a DPW for $L_n$ of size $|\A_n| |\E'_n|$. By choice of $L_n$, this DPW must be of size $2^{\Omega(n)}$, and since $\A_n$ is of size $n$, it follows that $\E'_n$, and hence $\E_n$, must be of size $2^{\Omega(n)}$. By a symmetric argument, every GFG NPW for $L(\C_n)$ must also be of size $2^{\Omega(n)}$.
\end{proof}

Informally, the language $L_n$ above describes a set of threads, of which at least one eventually satisfies a safety property. Then, the above construction can be understood as describing a property of reactive systems where, depending on the input, the system guarantees either that there is a thread that eventually satisfies a safety property, or that all threads satisfy a liveness (B\"uchi) property. The GFG alternating automaton can then be used to solve in polynomial time games with such languages as winning condition, for example in the context of synthesis: the product of the game arena and the alternating automaton for $L_n$ is a parity game with $3$ priorities with the same winner as the original game. In contrast, a DPW, GFG NPW and GFG UPW for the same language would all be exponentially larger.

\subsection{Deciding Half-GFGness}

In order to decide GFGness, it is enough to be able to decide the $\EGFG$ property on the automaton and its dual.
A natural first approach is therefore to study the complexity of deciding whether an APW is $\EGFG$.
Yet, we will show that already on finite words, this problem is \pspace-hard, while we conjecture that deciding GFGness is in \ptime.

\begin{lemma}
\label{lem:EGFG_PSPACE}
Deciding whether an AFA is $\EGFG$ is \pspace-hard.
\end{lemma}

\begin{proof}
We reduce from NFA universality: starting from an NFA $\A$, we build an AFA $\B$ based on the dual of $\A$, with an additional non-GFG choice to be resolved by Eve. This AFA $\B$ is $\EGFG$ if and only if $L(\B)=\emptyset$, which happens if and only if $L(\A)=\Sigma^*$.  We crucially use the fact that $\B$ is not necessarily $\AGFG$.

Let $\A$ be an NFA over an alphabet $\Sigma=\{a,b\}$ and $\bar\A$ its dual. We want to check whether $L(\A)=\Sigma^*$.
We build an AFA $\B$, as depicted below, by first making Eve guess the second letter. If her guess is wrong, the automaton proceeds to a rejecting sink state $\bot$. Otherwise, it proceeds to the initial state of $\bar{A}$. The size of $\B$ is linear in the size of $\A$.
\begin{center}
\includegraphics[scale=.5]{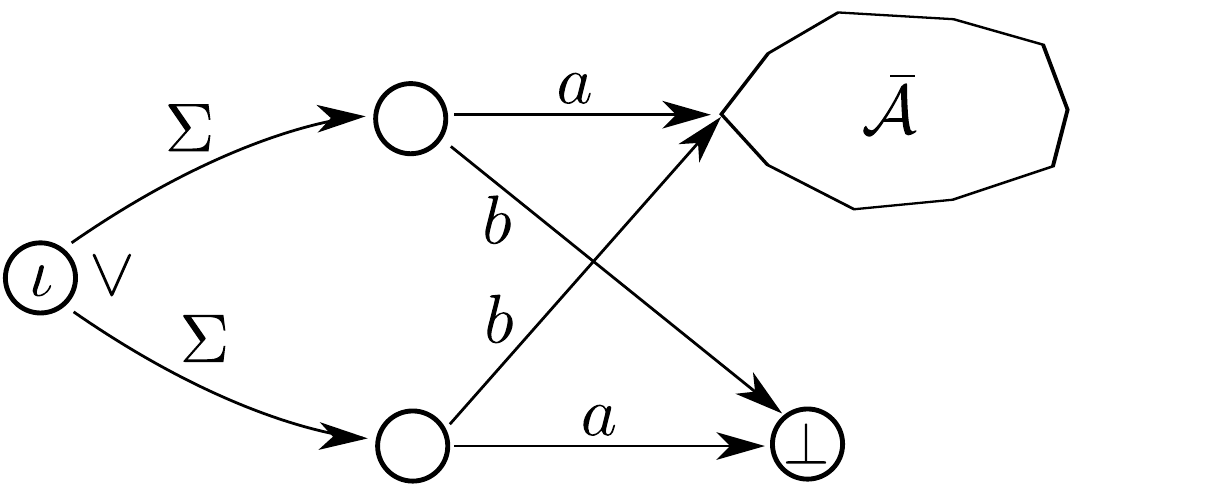}
\end{center}

If $L(\bar{A})=\emptyset$, then $L(\B)=\emptyset$, so $\B$ is trivially $\EGFG$.
However, if there is some $u\in L(\bar{\A})$, then Adam has a winning strategy in Eve's letter game on $\B$. This strategy consists of playing $a$, then playing the letter that brings Eve to $\bot$, and finally playing $u$. The resulting word is in $L(\B)=\Sigma^2L(\A)$, so this witnesses that $\B$ is not $\EGFG$.
We obtain that $L(\A)=\Sigma^* \Leftrightarrow L(\bar{A})=\emptyset\Leftrightarrow\B\text{ is }\EGFG$, which is the wanted reduction.
\end{proof}

For B\"uchi automata, and so in particular for finite words, we can give an \exptime algorithm for this problem.

\begin{lemma}
\label{lem:EGFG_EXPTIME}
Deciding whether an ABW is $\EGFG$ is in \exptime.
\end{lemma}

\begin{proof}
It is shown in~\cite[Lemma~23]{BL19} that removing alternation from an ABW $\A$ using the breakpoint construction~\cite{MH84} yields an NBW $\B$ such that if $\A$ is $\EGFG$ then $\B$ is GFG. Moreover, the converse also holds: if $\B$ is GFG then $\A$ is $\EGFG$, since playing Eve's letter game in $\B$ is more difficult for Eve than playing it in $\A$.
This means that starting from an ABW $\A$, we can build an exponential size NBW $\B$ via breakpoint construction, and test whether $\B$ is GFG via the algorithm from~\cite{BK18}, in time polynomial with respect to $\B$. Overall, this yields an \exptime algorithm deciding whether $\A$ is $\EGFG$.
\end{proof}



\section{Determinisation of Alternating GFG Parity Automata}
\label{sec:Determinisation}

In this section we provide a procedure that, given an~alternating GFG parity automaton, produces an equivalent deterministic parity automaton with single\=/exponentially many states. To do so, we first provide an alternation-removal procedure that preserves GFG status. Then, we apply this procedure to both the input automaton and its complement and use the GFG strategies in these two automata to determinise the input automaton. Our proofs, in~\cref{ap:Determinisation}, rely on some analysis of when GFG strategies can use the history of the word, rather than the full history of the play (which also includes the choices of how to resolve the nondeterminism and universality), and on the memoryless determinacy of parity games.

Our method for going from alternating to nondeterministic automata is similar to that of Dax and Klaedtke~\cite{DK08}: they take a nondeterministic automaton that recognises the universally\=/accepting words in $(\boxes_A)^\omega$ and add nondeterminism that upon reading a letter $\letter\in \Sigma$ chooses a box in $\boxes_A$ over \letter. Yet in our approach, in order to guarantee that the outcome preserves GFGness, the intermediate automaton is deterministic.

\subsection{Alternation Removal in GFG Parity Automata}


\begin{restatable}{theorem}{thmexpgfgdealt}
\label{thm:exp-gfg-dealt}
Consider an alternating parity automaton $\A$ with $n$ states and index $k$. There exists a nondeterministic parity automaton $\boxA$ with $2^{O(nk \log nk)}$ states that is equivalent to $\A$ such that if $\A$ is GFG then $\boxA$ is also GFG.
\end{restatable}

In~\cref{sec:deciding}, where we discuss decision procedures, we will show that $\boxA$ is GFG \textit{exactly} when $\A$ is $\EGFG$. For now, the rest of this section is devoted to the proof of~\cref{thm:exp-gfg-dealt}, of which a detailed version can be found in~\cref{app:alt-rem-gfg-rabin}.

\begin{restatable}{lemma}{lempropofd}
\label{lem:prop-of-d}
Consider an alternating parity automaton $\A$ with $n$ states and index $k$. Then there exists a deterministic parity automaton $\B$ with $2^{O(nk \log nk)}$ states over the alphabet $\boxes_\A$ that recognises the set of universally-accepting words for $\A$. If $\A$ is a B\"uchi automaton, then $\B$ can also be taken as B\"uchi, and in general the parity index of the automaton $\B$ is linear in the number of transitions~of~$\A$.
\end{restatable}

\ProofSketch
We first construct a nondeterministic parity (resp.\ coB\"uchi) automaton over the alphabet $\boxes_\A$ that recognises the complement of the set of universally-accepting words for $\A$. This automaton is easy to build: it guesses a path that is not accepting, and has the dual acceptance condition to $\A$.
We then obtain the automaton $\B$ by determinising and complementing this automaton. 
\end{proof}

\NotNeeded{
\begin{proof}
Notice that it is easy to construct a nondeterministic Streett (resp.\ parity) automaton $\S$ over the alphabet $\boxes_\A$ that recognises the complement of the set of universally-accepting words for $\A$---it is enough to guess a path that is not accepting, and have the acceptance condition that is the dual of $\A$'s condition. 
Formally, for an alternating Rabin (resp.\ parity) automaton $\A=\tuple{\Sigma, Q, \iota, \delta, \alpha}$, we define the nondeterministic Street (resp.\ parity) automaton $\S=\tuple{\boxes_\A, Q, \iota, \delta_\S, \overline{\alpha}}$, where $\overline{\alpha}$ is the dual of $\alpha$ and $\delta_\S$ is defined as follows.
For every states $q,q'\in Q$ and box $\beta\in\boxes_\A$, we have $q'\in\delta_\S(q,\beta)$ iff $\tuple{q,q'}\in\beta$.

Now, one can translate $\S$ to an equivalent deterministic parity automaton $\B'$ with $2^{O(nk \log nk)}$ states \cite{Pit07} (resp.\ $2^{O(n \log n)}$ states~\cite{CZ12,SV14}), and then complement the acceptance condition of $\B'$, getting the required automaton $\B$.

Since nondeterministic coB\"uchi automata can be determinised into deterministic coB\"uchi automata, if $\A$ is a B\"uchi automaton, so is $\B$.
\end{proof}
}


We now build the automaton $\boxA$ of \cref{thm:exp-gfg-dealt}. It is the same as the automaton $\B$ of \cref{lem:prop-of-d}, except that the alphabet is $\Sigma$ and the transition function is defined as follows: For every state $p$ of $\B$ and $\letter\in\Sigma$, we have $\delta_{\boxA}(p,\letter):= \bigcup_{\beta\in\boxes_{\tuple{\A,\letter}}} \delta_{\B}(p,\beta)$.

In other words, the automaton $\boxA$ reads a~letter $\letter$, nondeterministically guesses a~box $\beta\in\boxes_{\A,\letter}$, and follows the transition of $\B$ over $\beta$. Thus, the runs of $\boxA$ over a~word $w=w_0w_1w_2\dots\in\Sigma^\omega$ are in bijection with sequences of boxes $(\beta_i)_{i\in\Nat}$ such that $\beta_i\in\boxes_{\A,w_i}$ for all $i\in\Nat$.


Fix an infinite word $w\in\Sigma^\omega$. Our aim is to prove that $w\in L(\A)\Leftrightarrow w\in L(\boxA)$. 

\begin{restatable}{lemma}{lempositionalstrategiesandruns}
\label{lem:PositionalStrategiesAndRuns}
There exists a bijection between positional strategies of Eve in the acceptance game of $\A$ over $w$ and runs of $\boxA$ over $w$. Moreover, a strategy is winning if and only if the corresponding run is accepting.
Thus $L(\A)=L(\boxA)$.
\end{restatable}

\NotNeeded{
\begin{proof}
Consider a run of $\boxA$ over $w$, and observe that it corresponds to a sequence of boxes $\beta_0,\ldots$. Notice that each box $\beta_i$ corresponds to Eve's choices in $\A$ over $w_i$, and therefore provides a positional strategy for Eve in the one-step arena $R_{w_i}\times \A$. The sequence of these choices provides a positional strategy for Eve in $R_w\times \A$. 

Dually, given a positional strategy for Eve in $R_w\times \A$, one can extract a sequence of strategies for Eve in the one-step arenas $R_{w_i}\times \A$, and each of them corresponds to a box $\beta_i$. \cref{prop:choice-to-strat} shows that each path in $\beta_0,\ldots$ corresponds to a play consistent with the constructed strategy and vice versa: each play gives rise to a path.

Now, a run is accepting if and only if the sequence of boxes is universally accepting, which means exactly that all the plays consistent with the corresponding strategy are winning.
\end{proof}
}


\begin{remark}
The above alternation-removal procedure also extends to alternating Rabin automata but fails for alternating Streett automata~$\A$: since Streett games are not positionally determined for Eve, the acceptance game of $\A$ over a word $w$ is not positionally determined for Eve.
\end{remark}

\begin{restatable}{lemma}{lemgfgpreservation}
\label{lem:GFG-preservation}
For an alternating $\EGFG$ parity automaton $\A$, the automaton $\boxA$ is~GFG.
\end{restatable}

Intuitively, this is because the construction of $\boxA$ preserves the nondeterminism of $\A$. 

\NotNeeded{
\begin{proof}
Let $\stratE$ be a positional winning strategy for Eve in her expanded letter game for $\A$ (over the arena $R^\ast_{A,\Sigma}\times \A$). The proof is based on the construction of the function $\stratE'\colon \Sigma^+\rightarrow \boxes_\A$, see the paragraph before \cref{def:univ-acc-box}.

Consider the following way of resolving the nondeterminism of $\boxA$: after reading $w\in\Sigma^\ast$, when the next letter $\letter\in\Sigma$ is provided, the automaton moves to the state $\delta_\D(p,\beta_{w\letter})$ where $\beta_{w\letter}=\stratE'(w\letter)$. Consider an infinite word $w\in L(\A)$ and let $\beta_0,\ldots$ be the sequence of boxes used to construct the run of $\boxA$ over $w$. Lemma~\ref{lem:strat-to-win-boxes} implies that this sequence is universally accepting and therefore, the constructed run of $\D$ must also be accepting.
\end{proof}
}

\subsection{Single-Exponential Determinisation}
\label{ssec:exp-det-of-alt}

The aim of this section is to prove the following determinisation theorem; see~\cref{ap:ssec:exp-det-of-alt} for a detailed proof.

\begin{restatable}{theorem}{thmdet}
If $\A$ is an alternating parity GFG automaton then there exists a deterministic parity automaton $\D$ that recognises the same language and has size at most exponential in the size of $\A$. Moreover, the parity index of $\D$ is the same as that of $\A$.
\end{restatable}

\begin{remark}
\cref{thm:exp-gfg-dealt} and~\cite[Theorem~4]{BKKS13}, which uses an NRW-GFG and its complement NRW-GFG to obtain a DRW, together give an exponential deterministic parity automaton for $L(\A)$. However, the index of $\A$ might not be preserved. On the other hand, from~\cite[Theorem~19]{BL19} we know that there exists a~deterministic parity automaton equivalent to $\A$ with the same index, but it might have more than exponentially many states. Here we are able to guarantee both the preservation of the index and an exponential upper bound on the size of the deterministic automaton.
\end{remark}

Observe that \cref{thm:exp-gfg-dealt} can be applied both to $\A$ and its dual. Therefore, we can fix a~pair of nondeterministic GFG parity automata $\boxA$ and $\boxAco$ that recognise $L(\A)$ and $L(\A)^\mathrm{c}$ respectively and are both of size exponential in $\A$.
We use the automata $\A$, $\boxA$, and $\boxAco$ to construct two auxiliary games $G(\A)$ and $G'(\A)$ .\\

The game $G(\A)$ proceeds from a configuration consisting of a pair $(p,q)$ of states from $\boxAco$ and $\A$ respectively, starting from their initial states, as follows:
\begin{itemize}
\item Adam chooses a letter $\letter\in \Sigma$;
\item Eve chooses a transition $\trans{p}{\letter}{p'}$ in $\boxAco$;
\item Eve and Adam play on the one-step arena over $\letter$ from $q$ to a new state $q'$.
\end{itemize}

A play in $G(\A)$ consists of a run $\rho$ in $\boxAco$ and a path $\rho'$ in $\A$. It is winning for Eve if either $\rho$ is accepting in $\boxAco$ (in which case $w\notin L(\A)$), or $\rho'$ is accepting in $\A$.\\

If $\A$ is $\EGFG$ and $\boxAco$ is GFG, Eve has a winning strategy in $G(\A)$ consisting of building a run in $\boxAco$ using her GFG strategy in $\boxAco$ and a path in $\A$ using her $\EGFG$ strategy in $\A$. This guarantees that if $w\in L(\A)$ then the path in $\A$ is accepting, and otherwise the run in $\boxAco$ is accepting.

We then argue that as the winning condition of $G(\A)$ is a Rabin condition, Eve also has a winning strategy that is positional in $\A$, that is, which only depends on the history of the word and the current position. 
(Interestingly, the question of whether Eve can resolve the nondeterminism in a class of alternating GFG automata with only the knowledge of the word read so far does not tightly correspond to whether the acceptance condition of this class is memoryless. For example, it does hold for the generalised-B\"uchi condition, though it is not memoryless.) 
See~\cref{ap:Determinisation} for details.

\begin{remark}
There is some magic here: both the GFG strategies of Eve in $\A$ and in $\boxAco$ may require exponential memory, yet, when she needs to satisfy the disjunction of the two conditions, no more memory is needed. In a sense, the states of $\A$ provide the memory for $\boxAco$ and the states of $\boxAco$ provide the memory for $\A$.
\end{remark}

The game $G'(\A)$ is similar, except that Adam is given control of $\boxA$ and Eve is in charge of letters. This time Adam wins a play, consisting of a run of $\boxA$ and a path in $\A$, if either the path of $\A$ is rejecting  or the run of $\boxA$ is accepting.

Accordingly, if $\A$ is GFG, then he can win by using the $\EGFG$ strategy in $\boxA$ and the $\AGFG$ strategy in $\A$. Then if $w\in L(\A)$, the run in $\boxA$ is accepting, and otherwise the path of $\A$ is rejecting.
As before, he also has a positional winning strategy in $G'(\A)$.\\

\NotNeeded{
\todo{begin into appendix}
Now do the same with $\A$ and $\boxA$ for Adam: define $G$ as $\big(R_{E,\Sigma}\times \overline{\boxA}\big)\times \A$, where $\overline{\boxA}$ is the automaton $\boxA$ where the transitions are turned from nondeterministic to universal, i.e,\ we replace $\lor$ with $\land$.

Again, in a~round of $G$ from a~position $(v,p,q)$: Eve plays a letter $\letter$; Adam resolves nondeterminism of $\boxA$; then they both resolve the choices in $\A$. Let Adam win $G$ if either the play of $\A$ is rejecting or the run of $\boxA$ is accepting. Again we can ensure that Adam has a~winning strategy in $G$, because both automata are GFG: he uses the GFG strategy of $\boxA$ and the $\AGFG$ strategy over $\A$. If the word given by Eve belongs to $L(\A)$ then Adam wins by producing an accepting run of $\boxA$, otherwise he wins by refuting an accepting run of $\A$. Let $\stratA$ be his positional winning strategy in that game.

\todo{end}
}

We are now ready to build the deterministic automaton from a GFG APW $\A$, using positional winning strategies $\stratE$ and $\stratA$ for Eve and Adam in $G(\A)$ and $G'(\A)$, respectively.
 
Let $\D$ be the automaton with states of the form $(q,p_1,p_2)$, with $q$ a~state of $\A$, $p_1$ a~state of $\boxA$ and $p_2$ a state of $\boxAco$. 
A transition of $D$ over $\letter$ moves to $(q',p_1',p_2')$ such that moving from $(q,p_1)$ to $(q',p_1)$ is consistent with $\stratA$; and moving from $(q,p_2)$ to $(q',p_2')$ is consistent with $\stratE$.
The acceptance condition of $\D$ is inherited from $\A$.

\NotNeeded{
\todo{start appendix} 
When reading a~letter $\letter$ in such a~state, the following computations are performed:
\begin{enumerate}
\item We simulate the choices made by $\stratE$ in $G'(\A)$ upon obtaining $\letter$ from Adam. This way we know how to resolve nondeterminism of $\boxAco$ and what to do with disjunctions inside $\A$.
\item We simulate the choices made by $\stratA$ in $G(\A)$ upon obtaining $\letter$ from Eve. This way we know how to resolve nondeterminism in $\boxA$ and what to do with conjunctions of $\A$.
\item In the end we proceed to a new state of $\A$ and resolved nondeterminism of both $\boxA$ and $\boxAco$.
\end{enumerate}

Let the acceptance condition of $\D$ be inherited from $\A$.
\todo{end appendix}
}

\begin{restatable}{lemma}{lemdeteq}\label{lem:det-eq}
For a GFG APW $\A$ and $\D$ built as above, $L(\A)=L(\D)$.
\end{restatable}

\NotNeeded{
\todo{start appendix}

\begin{proof}
Take a word $w\in\Sigma^\omega$. First assume that $w\in L(\A)$. Eve cannot win a play of the game $G'$ with the letters played in $R_{A,\Sigma}$ coming from $w$ using by the first disjunct of her winning condition, since $L(\boxAco)=L(\bar\A)$. Thus, all the plays over $w$ consistent with her winning strategy $\stratE$ in $\G'$ must guarantee that the constructed path of $\A$ is accepting. Thus, the run of the automaton $\D$ over $w$ is accepting.

Now assume that $w\notin L(\A)$. Dually, no play of the game $G$ with the letters coming from $w$ can produce an accepting run of $\boxA$ over $w$. Thus, the strategy $\stratA$ guarantees that the sequence of visited states of $\A$ is rejecting. Thus, the run of $\D$ over $w$ must be rejecting.
\end{proof}

\todo{end appendix}
}

\begin{remark}
To extend this construction to an alternating GFG Rabin automaton~$\A$, we would need to remove alternations from both $\A$ and its dual while preserving GFGness. However, the dual is a Streett automaton, for which we cannot invoke positional determinacy.
\end{remark}

\section{Deciding GFGness of Alternating Automata}
\label{sec:deciding}

We use the development of the last section to show that deciding whether an APW is GFG is in \exptime. This matches the best known upper bound for the same problem on NPW. 


The main result of this section is the following theorem; its proof is in~\cref{app:exptime-gfg-alt}.

\begin{theorem}
\label{thm:exp-time-gfg-alt}
There exists an~\exptime algorithm that takes as input an alternating parity automaton $\A$ and decides whether $\A$ is GFG.
\end{theorem}

 The idea is to construct the (exponential size) NPWs $\boxA$ and $\boxAco$ for $L(\A)$ and $L(\A)^\mathrm{c}$ respectively, which are GFG if and only if $\A$ is $\EGFG$ and $\AGFG$ respectively. Then, it remains to check whether both are indeed GFG. Since we don't have a polynomial procedure to check this, instead, we will build a game which Eve wins if and only if \textit{both} are indeed GFG, and which we can solve in exponential time with respect to the size of $\A$.
 
 First, we observe the following reciprocal of \cref{lem:GFG-preservation}.

\begin{lemma}
\label{lem:GFG-B-to-A}
If $\boxA$ is GFG then $\A$ is $\EGFG$.
\end{lemma}

\begin{proof}
Assume that $\boxA$ is GFG and consider a~strategy witnessing this. Such a~strategy can be easily turned into a~function $\stratE'\colon \Sigma^+\rightarrow \boxes_\A$ that, given a~word $w\in L(\A)$ produces a~universally accepting word of boxes of $\A$. Now, due to the definition of a box, each such box defines a~positional strategy of Eve in the respective one-step game. This allows us to construct a~winning strategy of Eve in the letter game over $\A$.
\end{proof}

Thus, $\A$ is GFG if and only if both $\boxA$ and $\boxAco$ are GFG. To decide this, we consider a~game $G''$ where Adam plays letters and Eve produces runs of the automata $\boxA$ and $\boxAco$ in parallel. The winning condition of $G''$ requires that at least one of the constructed runs must be accepting.

Now, each sequence of letters given by Adam belongs either to the language of $\boxA$ or to $\boxAco$ and therefore, a winning strategy of Eve in $G''$ must comprise of two strategies witnessing GFGness of both $\boxA$ and $\boxAco$. Dually, if both $\boxA$ and $\boxAco$ are GFG then Eve wins $G''$ by playing the two strategies in parallel.

It remains to show that $G''$ is solvable in \exptime. Its winning condition is a disjunction of parity conditions, with index linear in the number of transitions of $\A$. This winning condition is recognised by a deterministic parity automaton of exponential size with polynomial index. To solve $G''$, we take its product with this deterministic automaton that recognises its winning condition, and solve the resulting parity game with an algorithm that is polynomial in the size of the game whenever, like here, the number of priorities is logarithmic in the size of the game, for instance~\cite{CJKLS17}. Details of this construction and its complexity are in~\cref{app:exptime-gfg-alt}.

\section{Conclusions}
\label{sec:conclusions}

The results obtained in this work shed new light on where alternating GFG automata resemble nondeterministic ones, and where they differ. 
Overall, our results show that allowing GFG alternations add succinctness without significantly increasing the complexity of determinisation nor decision procedures.

In particular, we show that alternating parity GFG automata can be exponentially more succinct than any equivalent nondeterministic GFG automata, yet this succinctness does not become double exponential when compared to deterministic automata, answering a question from~\cite{BL19}. Some further succinctness problems are left open here, such as the possibility of a doubly exponential gap between alternating GFG automata of stronger acceptance conditions and deterministic ones, as well as between $\EGFG$ parity automata and deterministic ones.

We also show that the interplay between the two players can be used to decide whether an automaton is GFG without deciding $\EGFG$ and $\AGFG$ separately, yielding an \exptime algorithm. This matches the current algorithms for deciding GFGness on non-deterministic automata. Bagnol and Kuperberg conjectured that GFGness is \ptime decidable for nondeterministic parity automata of fixed index~\cite{BK18}; we extend this conjecture to alternating automata.

It then becomes interesting to ask how to build an alternating automaton GFG. Indeed, Henzinger and Piterman~\cite{HP06} proposed a transformation of nondeterministic automata into GFG automata, which, despite in some cases leading to a deterministic automaton, is, conceptually, a much simpler procedure than determinisation. Indeed, in many examples of non-GFG automata, adding transitions suffices to obtain a GFG one. We leave finding such a procedure for alternating automata as future work.



\bibliography{gfg}

\clearpage

\noindent{\huge{\textbf{Appendix}}}
\appendix

\section{Appendix of \cref{sec:Preliminaries}}
\label{ap:Preliminaries}

In this section of the appendix we provide the remaining technical definitions from \cref{sec:Preliminaries} that are used in the proofs.

\newcommand{\GL}{C} 
\Subject{Games}
A~\emph{$\Sigma$-arena} is a directed (finite or infinite) graph with nodes (positions) split into $E$\=/labelled positions of Eve and $A$-labelled positions of Adam, where the edges (transitions) are labelled by elements of $\Sigma\sqcup\{\noLab\}$. The role of $\noLab$ is to mark edges that have no influence on the winner of a play, e.g., edges allowing players to resolve some Boolean formula.

We represent such an arena as $R=(V,X,V_E,V_A)$, where $V$ is its set of positions; $X\subseteq V\times \big(\Sigma\sqcup\{\noLab\}\big)\times V$ its transitions; $V_E\subseteq V$ the $E$\=/positions; and $V_A = V\setminus V_E$ the $A$\=/positions.

Notice that the definition allows more than one transition between a pair of positions (such transitions needs to have distinct labels). We will require that each infinite path contains infinitely many $\Sigma$\=/labelled transitions. An arena might be rooted at an initial position $v_\iota\in V$. We say that a position $v$ is \emph{terminal} if there is no outgoing transition from~$v$ (i.e.\ no element of $X$ of the form $(v,\letter,v')$). If we don't say that an arena is \emph{partial} then it is assumed that there are no terminal positions.

If $R$ is a (partial) $\Sigma$\=/arena and $V'\subseteq V$ is a set of positions, then $R\restr_{V'}$ is the \emph{sub-arena} of $R$ defined as the restriction of $R$ to the positions in $V'$, namely for $P\in\{E,A\}$, the $P$\=/positions of $R\restr_{V'}$ are $V'_P:= V_P\cap V'$, and its transitions are $X':= X\cap (V'\times (\Sigma\cup\{\epsilon\})\times V')$.
We say that two (partial) $\Sigma$-arenas 
$R=(V,X,V_E,V_A)$ and $R'=(V',X',V'_E,V'_A)$ are \emph{isomorphic} if there exists a bijection $i\colon V\to V'$ that preserves the membership in $V_P$/$V'_P$, for $P\in\{E,A\}$, and sets of transitions $X$/$X'$.

A~\emph{partial play} in $R$ is a path in $R$, i.e.,\ an element $\pi=v_0e_0v_1e_1\ldots$ of $V\cdot\big(X\cdot V)^\ast\cup\big(V\cdot X\big)^\omega$, where for every $i$ we have $e_i=(v_i,\letter_i, v_{i+1})$. Such a partial play is said to \emph{begin} in $v_0$. A~partial play is a \emph{play} if either it is infinite or the last position $v_i$ is terminal. 

A \emph{game} is a $\Sigma$-arena together with a winning condition $W\subseteq \Sigma^\omega$. An infinite play $\pi$ is said to be winning for Eve in the game if the sequence of $\Sigma$-labels $(\letter_i)_{i\in\Nat}$ of the transitions along $\pi$ form a word in $W$. Else $\pi$ is winning for Adam. Games with some class $X$ of winning conditions (e.g.,\ the parity condition) are called $X$ games (e.g.,\ parity games).

A \emph{strategy} for Eve (resp.\ Adam) is a function $\tau\colon V\cdot \big(X\cdot V\big)^*\rightarrow X$ that maps a \emph{history} $v_0e_0v_1\ldots e_{i-1}v_i$, i.e. a finite prefix of a play in $R$, to a transition $e_i$ whenever $v_i$ belongs to $V_E$ (resp. to $V_A$). A partial play $v_0e_0v_1e_1\dots$ agrees with a strategy $\tau$ for Eve (Adam) if whenever $v_i\in V_E$ (resp. in $V_A$), we have $e_i=\tau(v_0e_0v_1\ldots e_{i-1}v_i)$. A strategy for Eve (Adam) is winning from a position $v\in V$ if all plays beginning in $v$ that agree with it are winning for Eve (Adam). We say that a player wins the game from a position $v\in V$ if they have a winning strategy from $v$.
If the game is rooted at $v_\iota$, we say that a player wins the game if they win from $v_\iota$.

A strategy is \emph{positional} if its value depends only on the last position, i.e.,\ $\tau(v_0e_0\cdots e_{i-1}v_i)$ depends only on $v_i$. In that case the strategy of a player $P$ can be represented as a function $\tau\colon V_P\to X$.

We also define the notion of \emph{strategy with memory $M$} for player $P$. This is a tuple $(\sigma, M, m_0, \upd)$ where $M$ is a~set of \emph{memory states}; $m_0$ is an~\emph{initial memory state}; $\upd\colon M\times X\to M$ is an \emph{update function}, and $\sigma\colon M\times V_P\to X$ is a~strategy deciding which move should be played, depending only on the current memory state and on the current position. Along a play, the memory starts with $m_0$, and is updated along every transition according to $\upd$. The general notion of strategy corresponds to $M=(V\cdot X)^*$, and positional strategies correspond to $M$ being a singleton.
A player has a \emph{finite\=/memory winning strategy} if there exists a winning strategy using a finite memory set $M$.

\begin{proposition}
\label{lem:unfolding}
Let $G$ and $G'$ be two $\Sigma$-games with the same winning condition, such that the unfoldings of $G$ and $G'$ are isomorphic. Then Eve has a winning strategy in $G$ if and only if she has a winning strategy in $G'$.
\end{proposition}

\begin{proposition}[\cite{klarlund_progress_positional}]
\label{prop:RabinPositionalDeterminacy}
Rabin games are positionally determined for Eve. (If Eve has a~winning strategy then she has a positional winning strategy.)
\end{proposition}

\begin{definition}[Synchronised product]
\label{def:GameAutomataProd}
The \emph{synchronised product} $\pro{R}{\A}$ of a (partial) $\Sigma$-arena $R=(V,X,V_E,V_A)$ and an alternating automaton $\A=(\Sigma,Q,\iota,\delta,\alpha)$ with a set of transitions $T$ and labelling $\alpha: T \to \Gamma$ is a (partial) $\Sigma{\times}\Gamma$-arena defined as follows. Its set of positions is $(V\times Q) \cup (V\times Q\times\Sigma\times\widehat{\delta})$, and its transitions are defined by:

\begin{tabular}{ll}
$\big\langle (v,q), \noLab, (v',q,\letter,\delta(q,\letter))\big\rangle$
&for $(v,q)\in V\times Q$ and $\langle v, \letter, v'\rangle\in X$;\\
$\big\langle (v,q), \noLab, (v',q)\big\rangle$ 
& for $(v,q)\in V\times Q$ and $\langle v, \noLab, v'\rangle\in X$;\\
$\big\langle (v,q,\letter,b), \noLab, (v,q,\letter,b_i)\big\rangle$
&for $(v,q,\letter,b)\in V\times Q\times \Sigma\times \widehat{\delta}$\\
&with $b=b_1{\lor} b_2$ or $b=b_1{\land} b_2$ and $i=1,2$;\\
$\big\langle (v,q,\letter,q'), \big(\letter,\alpha(q,\letter,q')\big), (v,q')\big\rangle$
&for $(v,q,\letter,q')\in V\times Q\times \Sigma\times \widehat{\delta}$ with $q'\in Q$.
\end{tabular}

The positions belonging to Eve are of the form $(v,q)$ where $v\in V_E$ and of the form $(v,q,\letter,b_1{\lor}b_2)$. The remaining ones belong to Adam. If $R$ has an initial position $v_\iota$ then the initial position of the product is $(v_\iota,\iota)$.
\end{definition}

We implicitly assume that the arena only contains vertices that are reachable from $V\times Q$. (They need not be reachable from an initial position of $R$ and an initial state of $\A$, but from some position of $R$ and state of $\A$.)

We will sometimes consider longer products, like $(R\times \A_1)\times \A_2$, where $R$ is a~$\Sigma$\=/arena and both automata $\A_1$ and $\A_2$ are over the alphabet $\Sigma$. Assume that $\A_1$ and $\A_2$ have transitions labelled in sets $\Gamma_1$ and $\Gamma_2$ respectively. Notice that in that case the arena $R\times \A_1$ is formally a~$\Sigma\times\Gamma_1$\=/arena. Thus, to make the above formula precise, we treat the automaton $\A_2$ as an~automaton over the alphabet $\Sigma\times\Gamma_1$ and assume that it ignores the second component of the letters read.

\Subject{One-step arenas}
For a~letter $\letter\in\Sigma$, we denote by $R_\letter$ a~partial $\Sigma$\=/arena consisting of two vertices $v$ and $v'$ (it does not matter which of the players controls them), and one transition $\tuple{v, \letter, v'}$. Then, for an automaton $\A=(\Sigma,Q,\iota,\delta,\alpha)$, the product $R_\letter\times \A$ is a partial arena, in which the players should resolve their choices in the formulas $\delta(q,\letter)$ for all the possible states $q\in Q$. We call it the \emph{one\=/step arena} of $\A$ over $\letter$.
Such an arena contains one position of the form $(v,q)$ for each state $q\in Q$; a~set of non\=/terminal positions of the form $(v',q,\letter,\psi)$ for some $q\in Q$ and $\psi\in\widehat{\delta}$; and one terminal position of the form $(v',q)$ for each state $q\in Q$. (See \cref{fig:one-step-arena}.)

\Subject{Boxes}
\label{def:boxes}
In the later exposition, we will be interested in the combinatorial structure of possible strategies of Eve over one-step arenas $R_\letter\times \A$. For a positional strategy $\strat$ of Eve in a game of the form $R_\letter\times \A$, we define the \emph{box} of $\A$, $\letter$, and $\strat$, denoted by $\beta(\A,\letter,\strat)$, as the relation that is a subset of $Q\times \Sigma\times Q$ and contains a triple $(q,\letter,q')$ iff there exists a play in $R_\letter\times \A$ that is consistent with $\strat$, starting in $(v,q)$ and ending in $(v',q')$. We further define for every $\letter\in\Sigma$, the set $\boxes_{\tuple{\A,\letter}} = \{\beta(\A,\letter,\strat) \mid\text{$\strat$ is a positional strategy of Eve}\}$. Finally, let $\boxes_\A:=\bigcup_{\letter\in \Sigma}\boxes_{\tuple{\A,\letter}}$. Notice that $|\boxes_\A| \leq 2^{|Q\times \Sigma\times Q|}$.
When speaking of an arbitrary box, we mean any non-empty relation $\beta\subseteq Q\times\Sigma\times Q$ where all the letters $\letter$ appearing on the middle component are equal.

\cref{fig:boxes} represents $\boxes_{\tuple{\A,a}}$ for the automaton $\A$ of \cref{fig:one-step-arena}: Since there are two binary\=/choice positions of Eve in the corresponding one\=/step arena, there are four distinct positional strategies of Eve, which give the four possible boxes. They correspond to Eve choosing respectively LL, LR, RL, RR, where L stands for a left choice and R for a right choice in each of her two binary\=/choice positions.

\begin{proposition}
\label{prop:choice-to-strat}
Consider a letter $\letter$ and an automaton $\A$ with states $Q$ and transition function $\delta$. Then there is a bijection between $\boxes_{\tuple{\A,\letter}}$ and the positional strategies of Eve in the one-step arena of $\A$ and $\letter$.
\end{proposition}

\Subject{Acceptance of a word by an automaton}
We define the acceptance directly in terms of the model-checking (acceptance/membership) game, which happens to be exactly the product of the automaton with a path-like arena describing the input word. More precisely, given a word $w\in\Sigma^\omega$, the \emph{model-checking game} is defined as the product $\pro{R_w}{\A}$, where the arena~$R_w$ consists of an infinite path $\omega$, of which all positions belong to Eve (although it does not matter); the transitions are of the form $\langle i, w_i, i{+}1\rangle$; the initial position is $0$; and the winning condition is based on the winning condition of $\A$ (the $\Sigma$-component of the labels is ignored). We say that $\A$ \emph{accepts} $w$ if Eve has a winning strategy in the model-checking game $R_w\times \A$. The language of an automaton $\A$, denoted by $L(\A)$, is the set of words that it accepts (recognises).

Notice that for each $i\in\Nat$, the sub-arena of $\pro{R_w}{\A}$ with positions in
\[\big( \{i\}\times Q\big) \cup \big(\{i{+}1\} \times  Q{\times}\Sigma{\times}\BP(Q)\big)\cup \big(\{i{+}1\}\times Q\big)\]
is isomorphic to the one-step arena of $\A$ over $w_i$.

\NotNeeded{
\begin{figure}
\centering
\begin{tikzpicture}
\draw (-1, 1) edge[obrace] node{$R_w$} (1, 1);

\draw (3, 1) edge[obrace] node{$R_w\times \A$} (9, 1);

\foreach \x in {0, 1, 2, 3} {
	\node[eve, scale=0.9] (l\x) at (0, -2* \x) {$\x$};
	
	\foreach \s in {0,...,2} {
		\node[eve, scale=0.7] (v\x\s) at (4 + 2*\s, -2 * \x) {$(\x, q_{\s})$};
	}
}

\foreach \x/\l in {0/a, 1/b, 2/b} {
	\evalInt{\y}{\x+1}
	\draw[trans] (l\x) edge node[scale=0.8] {$\l$} (l\y);
	\node[inner sep=1pt] at (6, -2*\x - 1) {{\it one-step arena over \l}};
}

\node[dots] at (0, -7) {$\vdots$};
\node[dots] at (4, -7) {$\vdots$};
\node[dots] at (8, -7) {$\vdots$};

\end{tikzpicture}
\caption{The model-checking game defined as the product of the arena $R_w$ representing a given word $w$ and the automaton $\A$ of \cref{fig:one-step-arena}. The one-step arena of \cref{fig:one-step-arena} depicts the steps that correspond to the letter $a$.}
\label{fig:model-cheking-game}
\end{figure}
}

\section{Appendix of \cref{sec:Alternating-behaviour}}
\label{ap:Alternating}

\begin{definition}[A formalisation of \cref{def:LetterGames}]
Let $R_{A,\Sigma}$ be the $\Sigma$-arena consisting of a single position $v$ that belongs to Adam and the set of transitions $X$ of the form $\tuple{v,\letter,v}$ for each letter $\letter\in \Sigma$ (see \cref{fig:letter-giving-game}). The arena $R_{E,\Sigma}$ is the same except that $v$ belongs to Eve. Notice that the products $R_{A,\Sigma}\times \A$ and $R_{E,\Sigma}\times \A$ are both labelled by $\Sigma\times \Gamma$, where $\Sigma$ is the alphabet of $\A$ and $\Gamma$ is $\A$'s labelling, on top of which its acceptance condition is defined. Thus, the winning condition of games defined on these arenas can depend on a sequence of labels of the form $(\letter_i,\gamma_i)_{i\in\Nat}$. Then, \emph{Eve's letter game} is played over  $R_{A,\Sigma}\times \A$, where Eve wins if:
\[\text{$(\letter_i)_{i\in\Nat}\notin L(\A)$ or the sequence $(\gamma_i)_{i\in\Nat}$ satisfies the acceptance condition of $\A$.}\]
Dually, \emph{Adam's letter game} is played over $R_{E,\Sigma}\times \A$, where Adam wins if:
\[\text{$(\letter_i)_{i\in\Nat}\in L(\A)$ or the sequence $(\gamma_i)_{i\in\Nat}$ violates the acceptance condition of $\A$.}\]
\end{definition}

\section{Appendix of \cref{sec:Determinisation}}
\label{ap:Determinisation}

This section provides the technical details of the determinisation procedure in~\cref{sec:Determinisation}.
We start with some technical analysis of the types of histories needed to win letter games.

\subsection{Good for Games Automata: Required Histories}
\label{ap:sec:GfgDefinitions}

We begin by considering an~\emph{expanded} letter game. This will allow us to use a form of positional determinacy in letter games.

\begin{figure}[h]
\centering
\begin{tikzpicture}
\draw (-1.5, 1) edge[obrace] node{$R_{A,\Sigma}$} (1.5, 1);
\draw (3, 1) edge[obrace] node{$R^\ast_{A,\Sigma}$} (12, 1);

\node[adam, minimum size = 25pt] (v) at (0, -1) {$v$};

\draw[trans] (v) edge[out= 60,in=120, looseness=8] node[above] {$a$} (v);
\draw[trans] (v) edge[out=-30,in= 30, looseness=8] node[right] {$b$} (v);

\node[adam, minimum size = 20pt] (v0) at (7.5, +0.5) {$\epsilon$};
\node[adam, minimum size = 20pt] (v0a) at (6.5, -0.5) {$a$};
\node[adam, minimum size = 20pt] (v0b) at (8.5, -0.5) {$b$};

\node[adam, minimum size = 20pt, scale=0.8] (v0aa) at (6.0, -1.5) {$aa$};
\node[adam, minimum size = 20pt, scale=0.8] (v0ab) at (7.0, -1.5) {$ab$};

\node[adam, minimum size = 20pt, scale=0.8] (v0ba) at (8.0, -1.5) {$ba$};
\node[adam, minimum size = 20pt, scale=0.8] (v0bb) at (9.0, -1.5) {$bb$};

\draw[trans] (v0) edge node[above left]  {$a$} (v0a);
\draw[trans] (v0) edge node[above right] {$b$} (v0b);

\draw[trans] (v0a) edge node[scale=0.8, yshift=2pt, left]  {$a$} (v0aa);
\draw[trans] (v0a) edge node[scale=0.8, yshift=2pt, right] {$b$} (v0ab);

\draw[trans] (v0b) edge node[scale=0.8, yshift=2pt, left]  {$a$} (v0ba);
\draw[trans] (v0b) edge node[scale=0.8, yshift=2pt, right] {$b$} (v0bb);

\newcommand{\dotez}[1]{
\node at (#1, -2.0) {$\vdots$};
}

\dotez{5.5}
\dotez{6.5}
\dotez{7.5}
\dotez{8.5}
\dotez{9.5}
\end{tikzpicture}
\caption{The arenas $R_{A,\Sigma}$ and $R^\ast_{A,\Sigma}$, allowing Adam to choose an arbitrary sequence of letters. We define Eve's letter game for an automaton $\A$ over the product of $R_{A,\Sigma} \times \A$, and her expanded letter game over $R^\ast_{A,\Sigma} \times \A$; 
She wins a play if the word generated by Adam is not in $L(\A)$ or the path generated by her (resolving $\A$'s nondeterminism) and by Adam (resolving $\A$'s universality) satisfies $\A$'s acceptance condition.}
\label{fig:letter-giving-game}
\end{figure}
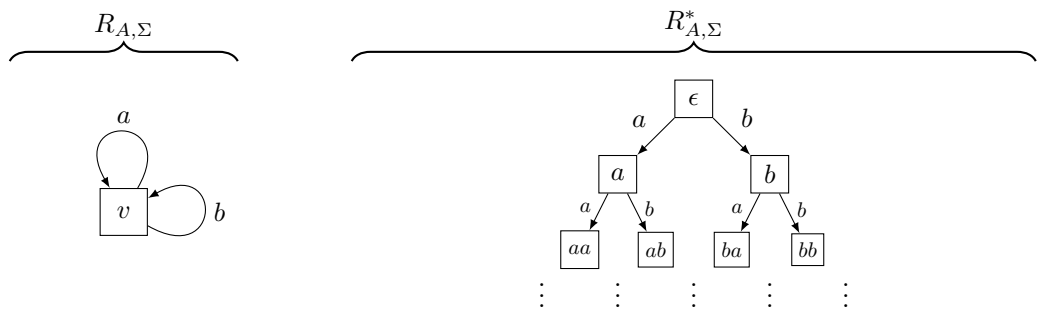

\Subject{Expanded letter games}
The definition of the letter games (\cref{def:LetterGames}) has the important advantage of being defined over a finite-arena. Yet, as a result, these games generally do not allow for positional determinacy. 

We provide below an expanded variant of the letter game that will have same unfolding as the original one, while being defined over an infinite arena. This will allow Eve to have positional determinacy in these games for Rabin automata.

Let $R^\ast_{A,\Sigma}$ be the $\Sigma$-arena with the set of positions $V=\Sigma^\ast$, all belonging to Adam, and the set of transitions $X$ of the form $\langle w, a, w\cdot \letter\rangle$ for each word $w\in\Sigma^\ast$ and letter $\letter\in\Sigma$ (see \cref{fig:letter-giving-game}). The initial position of this arena is $\epsilon$. The arena $R^\ast_{E,\Sigma}$ is the same, except that all the positions belong to Eve. We define \emph{Eve's expanded letter game} over $R^\ast_{A,\Sigma}\times \A$ and \emph{Adam's expanded letter game} over $R^\ast_{E,\Sigma}\times \A$ with the same winning conditions as in their (non expanded) variants.
The following follows directly from \cref{lem:unfolding}.

\begin{proposition}
For every automaton $\A$, the expanded letter games for $\A$ have the same winners as the (standard) letter games for $\A$.
\end{proposition}

\Subject{History Requirement}
Although  $R^\ast_{A,\Sigma}$ and $R^\ast_{E,\Sigma}$ are trees, 
the arenas of the expanded letter games are directed acyclic graphs as there can exist two distinct paths from the initial position to a~given position $(w,q)$. Thus, a priori, a winning strategy of a~player of such a game might need some history of a~play. However, as expressed by the following theorem, it is not the case.

\begin{theorem}\label{thm:ExpandedGamePositional}
If $\A$ is a Rabin (or parity) automaton then Eve's expanded letter game is positionally determined for Eve.
\end{theorem}

\begin{proof}
We will show that the winning condition of this game can be represented as a Rabin condition and invoke \cref{prop:RabinPositionalDeterminacy}. Let $\D'$ be a deterministic parity automaton recognising the complement of the language $L(\A)$, for a Rabin automaton $\A$ over the alphabet $\Sigma$. Let $\Omega\colon Q_{\D'}\to\Nat$ be the priority assignment of $\D'$ (without loss of generality we can assume that the states of $\D'$ bear priorities). 

Consider the arena $R'^\ast_{A,\Sigma}$ that is derived from the arena $R^\ast_{A,\Sigma}$ by adding to transitions priorities according to the deterministic of runs of $D'$, that is, by changing the labelling of every transition $\langle w, \letter, w\cdot \letter\rangle$ to $\langle w, (\Omega_{\D'}(q),\letter), w\cdot \letter\rangle$, where $\letter\in\Sigma$ and $q$ is the state of $\D'$ reached after reading the word $w$ from the initial state of $\D'$.

Consider the product $R'^\ast_{A,\Sigma}\times \A$, in which for $\A's$ transitions we ignore these additional labels. The labels of that product are now of the form $(\ell,\letter,\gamma)$, where $\ell\in\Nat$ is a priority of $\D'$, $\letter\in\Sigma$, and $\gamma\in\Gamma$ is a label of $\A$. Notice that when one forgets about the first coordinate of the label, the game is equal to $R^\ast_{A,\Sigma}\times \A$. Moreover, given a~sequence of labels $(\ell_i,\letter_i,\gamma_i)_{i\in\N}$, by the choice of $\D'$, we know that $(a_i)_{i\in\N}\notin L(\A)$ if and only if the sequence $(\ell_i)_{i\in\N}$ satisfies the parity condition.

Define the game $\G$ over $R'^\ast_{A,\Sigma}\times \A$, in which Eve wins a play labelled by $(\ell_i,\letter_i,\gamma_i)_{i\in\Nat}$ if
\[\text{$(\ell_i)_{i\in\Nat}$ satisfies the parity condition of $\D'$ or $(\gamma_i)_{i\in\Nat}$ satisfies the Rabin condition of $\A$.}\]
Notice that both disjuncts above can be written as Rabin conditions and therefore $\G$ is positionally determined for Eve. Moreover, the choice of $\D'$ guarantees that the new winning condition is equivalent to Eve's condition in her expanded letter game on $\A$---the same plays are winning for Eve in $\G$ and her expanded letter game on $\A$. Since the structure of the game is also preserved, it means that Eve's expanded letter game on $\A$ is positionally determined for Eve.
\end{proof}

\begin{remark}
Dually, Adam's expanded letter game for a Streett automaton is positionally determined for Adam.
\end{remark}

As a consequence  of \cref{thm:ExpandedGamePositional}, for alternating $\EGFG$ Rabin automata, a strategy for Eve to resolve the nondeterminism may ignore the history of the play, and only consider the history of the word read, as is the case for nondeterministic GFG automata.

We will now argue that Eve's positional strategy $\stratE$ in the expanded letter game on an alternating automaton $\A$ can be represented as a function $\stratE'\colon \Sigma^+\rightarrow \boxes_\A$ that assigns to each word $w\letter$ a box $\beta_{w\letter}\in\boxes_{\tuple{\A,\letter}}$. Indeed, let $(w\cdot\letter)\in\Sigma^+$ and let $V_{w\letter}$ be the set of positions of $R^\ast_{A,\Sigma}\times \A$ of the form $(w,q)$, $(w\cdot \letter, q,\letter,b)$, or $(w\cdot \letter, q)$ for $q\in Q$ and $b\in \BP(Q)$. Observe that the partial arena of $R^\ast_{A,\Sigma}\times \A$ restricted to $V_{w\letter}$ is isomorphic to the one-step arena $R_\letter\times \A$. Thus, $\stratE$ provides a positional strategy over this arena, which by \cref{prop:choice-to-strat} can be encoded as a box $\beta_{w\letter}$. More formally, let $\beta_{w\letter}$ contain $(q,a,q')$, if there is a~play consistent with $\stratE$ that visits both the positions $(w,q)$ and then $(w\letter,q')$.

Then, in the next lemma we show that if $\stratE$ is also winning, then the sequences of boxes $\beta_{w\letter}$ only has accepting paths.

\begin{definition}
\label{def:univ-acc-box}
Consider an automaton $\A$ with states $Q$ and initial state $\iota$, and an infinite word $u=\beta_0,\beta_1,\ldots\in(\boxes_\A)^\omega$. We say that a sequence of transitions $\rho=(q_i,\letter_i,q_{i+1})_{i\in\Nat}$ is a \emph{path of $u$} if $q_0=\iota$ and for every $i\in\Nat$, we have $(q_i,\letter_i q_{i+1}) \in \beta_i$. The word $u$ is \emph{universally accepting for $\A$} if each of its paths satisfies the acceptance condition of $\A$.
\end{definition}

\begin{lemma}
\label{lem:strat-to-win-boxes}
Given an alternating $\EGFG$ Rabin automaton $\A$, there is a positional strategy $\stratE$ in her expanded letter game on $\A$ such that
for every word $w\in L(\A)$ the sequence of boxes $u=\beta_0,\beta_1,\ldots\in(\boxes_\A)^\omega$ defined as $\beta_i=\stratE(w\restr_{i+1})$ is universally accepting for $\A$.
\end{lemma}

\begin{proof}
Consider words $w$ and $u$ as above. Let $\rho=(q_i,\letter_i,q_{i+1})_{i\in\Nat}$ be a path of $u$. Since the strategy $\stratE$ is positional, the definition of $\stratE(w\restr_{i+1})$ implies that there exists a~single play of the expanded letter game that visits all the positions of the form $(w\restr_i,q_i)$ for $i\in\N$. Since $w\in L(\A)$, the winning condition of the expanded letter game guarantees that the path $\rho$ must be accepting.
\end{proof}

Observe that the above arguments do not hold for alternating GFG Streett automata: Since Streett games are not positionally determined for Eve, Eve's expanded letter game for a Streett automaton is not positionally determined for Eve (an analogous of \cref{thm:ExpandedGamePositional} does not hold). Furthermore, we provide in \cref{fig:StreetGFGHistory} an example of an alternating GFG Streett automaton, in which Eve cannot resolve her nondeterminism only according to the history of the word read.

\begin{proposition}
Consider an $\EGFG$ alternating Streett automaton $\A$ with transition conditions in DNF. Then Eve might not have a strategy~$\stratE\colon \Sigma^+\rightarrow \boxes_\A$ satisfying Lemma~\ref{lem:strat-to-win-boxes}.
\end{proposition}

\begin{proof}
Consider the ASW depicted in \cref{fig:StreetGFGHistory}. It is $\EGFG$, as witnessed by the strategy that chooses the transition $t_4$ in $q_3$ if the last visited state was $q_1$ and $t_5$ otherwise. Yet, there is no strategy that only remembers the word read so far, as this only gives the length of the word, and cannot help in determining whether the path visited $q_1$ or $q_2$.
\end{proof}

\begin{figure}
\centering
\begin{tikzpicture}

   \node (b) at (0, 0.75) {};
	\node[state, scale=0.9] (q0) at (0, 0) {$q_0$};
	\node[state, scale=0.9] (q1) at ($(q0)+(-1,-1)$) {$q_1$};
	\node[state, scale=0.9] (q2) at ($(q0)+(+1,-1)$) {$q_2$};
	\node[state, scale=0.9] (q3) at ($(q0)+(0,-1.75)$) {$q_3$};
	\node[state, scale=0.9] (q4) at ($(q3)+(-1,-1)$) {$q_4$};
    \node[state, scale=0.9] (q5) at ($(q3)+(+1,-1)$) {$q_5$};

\node [bool] (b0) at ($(q0)+(0,-0.5)$) {$\land$};
\node [bool] (b3) at ($(q3)+(0,-0.5)$) {$\lor$};

\draw[trans] (b) edge (q0);
\draw[trans] (b0) edge node[scale=0.8, below] {$t_{1}$} (q1);
\draw[trans] (b0) edge node[scale=0.8, below] {$t_{2}$} (q2);
\draw[trans] (q1) edge  (q3);
\draw[trans] (q2) edge  (q3);
\draw[trans] (b3) edge node[scale=0.8, below] {$t_4$} (q4);
\draw[trans] (b3) edge node[scale=0.8, below] {$t_5$} (q5);
\draw[trans] (q4) edge[bend left=70] (q0);
\draw[trans] (q5) edge[bend right=70] (q0);

\node (AC) at (6,-0.5)  [align=left]{Acceptance condition:\\(Finitely often $t_1$ or Infinitely often $t_4$) and\\(Finitely often $t_2$ or Infinitely often $t_5$)};

\end{tikzpicture}
\caption{A GFG ASW over  a singleton alphabet, for which Eve's $\EGFG$ strategy cannot only remember the prefix of the word read so far, but also some history about the visited states.}
\label{fig:StreetGFGHistory}
\end{figure}



\subsection{Alternation Removal in GFG Parity Automata}
\label{app:alt-rem-gfg-rabin}

This section presents the proof of the following theorem:

\thmexpgfgdealt*


\lempropofd*

\begin{proof}
Notice that it is easy to construct a nondeterministic parity automaton $\S$ over the alphabet $\boxes_\A$ that recognises the complement of the set of universally-accepting words for $\A$---it is enough to guess a path that is not accepting, and have the acceptance condition that is the dual of $\A$'s condition. If $\A$ is a B\"uchi automaton, then $\S$ is a coB\"uchi one.
Formally, for an alternating parity (resp.\ B\"uchi) automaton $\A=\tuple{\Sigma, Q, \iota, \delta, \alpha}$, we define the nondeterministic parity (resp.\ coB\"uchi) automaton $\S=\tuple{\boxes_\A, Q, \iota, \delta_\S, \overline{\alpha}}$, where $\overline{\alpha}$ is the dual of $\alpha$ and $\delta_\S$ is defined as follows.
For every states $q,q'\in Q$ and box $\beta\in\boxes_\A$, we have $q'\in\delta_\S(q,\beta)$ iff $\trans{q}{\letter}{q'}\in\beta$ for some $\letter$.

Now, one can translate $\S$ to an equivalent deterministic parity automaton $\B'$ with $2^{O(nk \log nk)}$ states \cite{Pit07} 
and then complement the acceptance condition of $\B'$, getting the required automaton $\B$.

Since nondeterministic coB\"uchi automata can be determinised into deterministic coB\"uchi automata, if $\A$ is B\"uchi, so is $\B$. In general, the parity index of the automaton $\B$ is linear in the number of transitions~of~$\A$.
\end{proof}

We now proceed to the construction of the automaton $\boxA$ of \cref{thm:exp-gfg-dealt}. It is the same as the automaton $\B$ of \cref{lem:prop-of-d}, except that the alphabet is $\Sigma$ and the transition function is defined as follows: For every state $p$ of $\boxA$ and $\letter\in\Sigma$, we have $\delta_{\boxA}(p,\letter):= \cup_{\beta\in\boxes_{\tuple{\A,\letter}}} \delta_{\B}(p,\beta)$.

In other words, the automaton $\boxA$ reads a~letter $\letter$, nondeterministically guesses a~box $\beta\in\boxes_{\A,\letter}$, and follows the transition of $\B$ over $\beta$. Thus, the runs of $\boxA$ over a~word $w\in\Sigma^\omega$ are in bijection between sequences of boxes $(\beta_i)_{i\in\Nat}$ such that $\beta_i\in\boxes_{\A,w_i}$ for $i\in\N$.


Fix an infinite word $w\in\Sigma^\omega$. Our aim is to prove that $w\in L(\A)\Leftrightarrow w\in L(\boxA)$. 

\lempositionalstrategiesandruns*

\begin{proof}
Consider a run of $\boxA$ over $w$, and observe that it corresponds to a sequence of boxes $\beta_0,\ldots$. Notice that each box $\beta_i$ corresponds to Eve's choices in $\A$ over $w_i$, and therefore provides a positional strategy for Eve in the one-step arena $R_{w_i}\times \A$. The sequence of these choices provides a positional strategy for Eve in $R_w\times \A$. 

Dually, given a positional strategy for Eve in $R_w\times \A$, one can extract a sequence of strategies for Eve in the one-step arenas $R_{w_i}\times \A$, and each of them corresponds to a box $\beta_i$. \cref{prop:choice-to-strat} shows that each path in $\beta_0,\ldots$ corresponds to a play consistent with the constructed strategy and vice versa: each play gives rise to a path.

Now, a run is accepting if and only if the sequence of boxes is universally accepting, which means exactly that all the plays consistent with the corresponding strategy are winning.
\end{proof}

We now show that the automaton $\boxA$ is also GFG.

\NotNeeded{
\begin{remark}
The above alternation-removal procedure does not work for an alternating Streett automaton $\A$, even if ignoring the issue of GFGness: Since Streett games are not positionally determined for Eve, the acceptance game of $\A$ over a word $w$ is not positionally determined for Eve.
That is, a winning strategy for Eve in this game might make different choices at a position $\tuple{q,w_i}$, depending on the path leading to this position. Such a strategy does not correspond to a sequence of boxes, and therefore the analogous of \cref{lem:PositionalStrategiesAndRuns} does not hold.
\end{remark}
}

\lemgfgpreservation*

\begin{proof}
Let $\stratE$ be a positional winning strategy for Eve in her expanded letter game for $\A$ (over the arena $R^\ast_{A,\Sigma}\times \A$). The proof is based on the construction of the function $\stratE'\colon \Sigma^+\rightarrow \boxes_\A$, see the paragraph before \cref{def:univ-acc-box}.

Consider the following way of resolving the nondeterminism of $\boxA$: after reading $w\in\Sigma^\ast$, when the next letter $\letter\in\Sigma$ is provided, the automaton moves to the state $\delta_\B(p,\beta_{w\letter})$ where $\beta_{w\letter}=\stratE'(w\letter)$. Consider an infinite word $w\in L(\A)$ and let $\beta_0,\ldots$ be the sequence of boxes used to construct the run of $\boxA$ over $w$. Lemma~\ref{lem:strat-to-win-boxes} implies that this sequence is universally accepting and therefore, the constructed run of $\B$ must also be accepting.
\end{proof}

\subsection{Single-Exponential Determinisation of Alternating Parity GFG Automata}
\label{ap:ssec:exp-det-of-alt}

The aim of this section is to prove the following determinisation theorem.

\thmdet*

First consider the synchronised product $R_{A,\Sigma}\times \boxAco$, which is a~game with labels of the form $\Sigma\times \Gamma_{\boxAco}$, where $\Gamma_{\boxAco}$ is the parity condition of $\boxAco$. Now, we can treat the automaton $\A$ as an~automaton over the alphabet $\Sigma\times \Gamma_{\boxAco}$ that just ignores the second component of the given letter. Thus, we can define a~game $G(\A)= \big(R_{A,\Sigma}\times \boxAco\big)\times \A$.

Notice that $G(\A)$ is naturally divided into rounds, between two consecutive positions of the form $(v,p,q)$, where $v$ is the unique position of $R_{A,\Sigma}$, $p$ is a~state of $\boxAco$, and $q$ is a~state of $\A$. Such a~round, starting in $(v,p,q)$, consists of first Adam choosing a~letter $\letter$; then Eve resolving nondeterminism of $\boxAco$ from $p$ over $\letter$; and then both players playing the game corresponding to the transition condition $\delta(q,\letter)$ of $\A$.

Let the winning condition of $G(\A)$ say that either the sequence of transitions of $\boxAco$ is accepting or the sequence of transitions of $\A$ is accepting. Since $\A$ is $\EGFG$ and $\boxAco$ is GFG, we know that Eve has a~winning strategy in $G(\A)$: she just plays her GFG strategies in both automata and is guaranteed to win whether the word produced by Adam is in $L(\A)$ or $L(\boxAco)$.

As the winning condition of $G(\A)$ is a disjunction of two Rabin conditions, Eve has a positional winning strategy. Fix such a strategy $\stratE$.

\NotNeeded{
\begin{remark}
There is some magic here, as both the GFG strategies of Eve in $\A$ and in $\boxAco$ may require exponential memory. Yet, when she needs to satisfy the disjunction of the two conditions, no more memory is needed. In a sense, the states of $\A$ provide the memory for $\boxAco$ and the states of $\boxAco$ provide the memory for $\A$.  (cf.\  
\cite[Theorem 4]{BKKS13}).
\end{remark}

The second auxiliary game $G'(A)$ is similar, except that Adam is given control of $\boxA$ and Eve is in charge of letters. That is, in $G'(A)$, the configurations are again pairs $(p,q)$ of states from $\boxA$ and $\A$ respectively, and the game proceeds, starting from the initial states, as follows:
\begin{itemize}
\item Eve chooses a letter $\letter\in \Sigma$;
\item Adam chooses a transition $\trans{p}{\letter}{p'}$ in $\boxA$;
\item Adam and Eve play on the one-step arena over $\letter$ from $q$ to a new state $q'$.
\end{itemize}

This time Adam wins a play consisting of a run $\rho$ of $\boxA$ and a path $\rho'$ in $\A$ if either the path is rejecting of $\A$ is rejecting  or the run of $\boxA$ is accepting.

Accordingly, if $\A$ is GFG, then he can win by using the GFG strategy in $\boxA$ and the $\AGFG$ strategy in $\A$. Then if $w\in L(\A)$, the run in $\boxA$ is accepting, and otherwise the path of $\A$ is rejecting.
As before, he also has a positional winning strategy in $G'$.

}

Now do the same with $\A$ and $\boxA$ for Adam: define $G'(\A)$ as $\big(R_{E,\Sigma}\times \overline{\boxA}\big)\times \A$, where $\overline{\boxA}$ is the automaton $\boxA$ where the transitions are turned from nondeterministic to universal, i.e,\ we replace $\lor$ with $\land$.

Again, in a~round of $G'(\A)$ from a~position $(v,p,q)$: Eve plays a letter $\letter$; Adam resolves nondeterminism of $\boxA$ (i.e., the universality in its dual); then they both resolve the choices in $\A$. Let Adam win $G$ if either the play of $\A$ is rejecting or the run of $\boxA$ is accepting. Again we can ensure that Adam has a~winning strategy in $G'(\A)$, because both automata are GFG: he uses the GFG strategy of $\boxA$ and the $\AGFG$ strategy over $\A$. If the word given by Eve belongs to $L(\A)$ then Adam wins by producing an accepting run of $\boxA$, otherwise he wins by refuting an accepting run of $\A$. Let $\stratA$ be his positional winning strategy in that game.

We are now ready to build the deterministic automaton from a GFG APW $\A$, using positional winning strategies $\stratE$ and $\stratA$ for Eve and Adam in $G'(\A)$ and $G(\A)$, respectively.
 
Let $\D$ be the automaton with states of the form $(q,p_1,p_2)$, with $q$ a~state of $\A$, $p$ a~state of $\boxA$ and $p'$ a state of $\boxAco$. 
A transition of $D$ over $\letter$ moves to $(q',p_1',p_2')$ such that $((q,p_1),(q',p_1))$ is consistent with $\stratA$ and $((q,p_2'),(q',p_2'))$ is consistent with $\stratE$. In other words, 
when reading a~letter $\letter$ in such a~state, the following computations are performed:
\begin{enumerate}
\item We simulate the choices made by $\stratE$ in $G'(\A)$ upon obtaining $\letter$ from Adam. This way we know how to resolve nondeterminism of $\boxAco$ and what to do with disjunctions inside~$\A$.
\item We simulate the choices made by $\stratA$ in $G(\A)$ upon obtaining $\letter$ from Eve. This way we know how to resolve nondeterminism in $\boxA$ and what to do with conjunctions of $\A$.
\item In the end we proceed to a new state of $\A$ and resolved nondeterminism of both $\boxA$ and $\boxAco$.
\end{enumerate}

The acceptance condition of $\D$ is inherited from $\A$.

\lemdeteq*

\begin{proof}
Take a word $w\in\Sigma^\omega$. First assume that $w\in L(\A)$. Eve cannot win a play of the game $G$ with the letters played in $R_{A,\Sigma}$ coming from $w$ using by the first disjunct of her winning condition, since $L(\boxAco)=L(\bar\A)$. Thus, all the plays over $w$ consistent with her winning strategy $\stratE$ in $\G'$ must guarantee that the constructed path of $\A$ is accepting. Thus, the run of the automaton $\D$ over $w$ is accepting.

Now assume that $w\notin L(\A)$. Dually, no play of the game $G'$ with the letters coming from $w$ can produce an accepting run of $\boxA$ over $w$. Thus, the strategy $\stratA$ guarantees that the sequence of visited states of $\A$ is rejecting. Thus, the run of $\D$ over $w$ must be rejecting.
\end{proof}


\section{Appendix of \cref{sec:deciding}}
\label{ap:deciding}

\subsection{Proof of \cref{thm:exp-time-gfg-alt}}
\label{app:exptime-gfg-alt}

Our aim is to provide an \exptime{} algorithm for deciding if a given alternating parity automaton is GFG.

Recall the construction of the two nondeterministic parity automata $\boxA$ and $\boxAco$ for $L(\A)$ and $L(\A)^c$ respectively, as defined in \cref{ssec:exp-det-of-alt}. We will use these automata to design a~game characterising the fact that $\A$ is both $\EGFG$ and $\AGFG$, i.e,\ $\A$ is just GFG.

Recall that the automata $\boxA$ and $\boxAco$ have exponential number of states in the number of states of $\A$. However, due to \cref{lem:prop-of-d} their parity index is linear in the number of transitions of $\A$. Consider the game $G''=\big(R_{A,\Sigma}\times \boxA\big)\times \boxAco$, i.e,\ the game where Adam plays a~letter and Eve replies with two boxes, one of $\A$ and the other of $\overline{\A}$. Let the winning condition of that game for Eve say that either of the runs of $\boxA$ or $\boxAco$ must be accepting.

\begin{lemma}
Eve has a~winning strategy in $G''$ if and only if $\A$ is GFG.
\end{lemma}

\begin{proof}
Clearly if $\A$ is GFG then both $\boxA$ and $\boxAco$ are GFG as nondeterministic automata. Therefore, one can use strategies witnessing their GFGness to construct a single strategy for Eve in $G''$. This strategy must be winning, because each word proposed by Adam either belongs to $L(\boxA)=L(\A)$ or to $L(\boxAco)=L(\A)^c$.

Now assume that Eve has a~winning strategy in $G''$. This strategy consists of two components: one is a strategy in $R_{A,\Sigma}\times \boxA$ and the other in $R_{A,\Sigma}\times \boxAco$. By the fact that the languages of $\boxA$ and $\boxAco$ are disjoint, the above components are in fact winning strategies in the letter games for $\boxA$ and $\boxAco$ respectively. Thus, by \cref{lem:GFG-B-to-A} we know that $\A$ is both $\EGFG$ and $\AGFG$.
\end{proof}

What remains is to show how to solve the game $G''$ in \exptime{}. Let $n$ be the size of the automaton $\A$. Our aim is to turn it into a~parity game of size exponential in $n$ but with a~number of priorities polynomial in $n$. Then, by invoking for instance \cite{CJKLS17}, we know that such a~game can be solved in \exptime.

\begin{lemma}
\label{lem:aut-for-disjunction}
Let $\Gamma=\{0,\ldots,N\}$ be a~set of priorities. Then, there exists a~deterministic parity automaton of size exponential in $N$, with a~number of priorities polynomial in $N$ that recognises the language $L$ of words $w\in(\Gamma\times\Gamma)^\omega$ that satisfy the parity condition on at least one coordinate.
\end{lemma}

\begin{proof}
It is a~rather standard construction. One possibility is to design a~nondeterministic B\"uchi automaton for $L$ with $N^2$ states. Then, the standard determinisation procedure \cite{Pit07} applied to this automaton gives a~deterministic parity automaton as in the statement.
\end{proof}

Therefore, we conclude the proof of \cref{thm:exp-time-gfg-alt} by taking a~product of the game $G''$ with the automaton from \cref{lem:aut-for-disjunction} and then solving the resulting parity game.

\end{document}